\documentclass[pra,twocolumn,showpacs,nofootinbib,superscriptaddress,notitlepage]{revtex4-1}
\usepackage{amsmath,empheq}
\usepackage{amssymb,bm}
\usepackage{amsthm,color,xcolor,dsfont}
\usepackage{graphicx} 
\usepackage{xcolor}
\usepackage{epstopdf}
\usepackage[colorlinks=true, hyperindex, breaklinks, linkcolor=blue, urlcolor=blue, citecolor=blue]{hyperref} 
\usepackage{ulem}
\usepackage{caption}
\usepackage{subcaption}
\usepackage{cleveref}
\newcommand{\CC}{\mathcal{C}}

\newcommand{\ket}[1]{|#1\rangle} 
 
\newcommand{\oline}[1]{\overline{#1}}

 \newtheorem{lemma}{Lemma}
\newtheorem{corollary}{Corollary} 
\newtheorem{claim}{Claim}

\begin{document}

\title{Stacked codes: universal fault-tolerant quantum computation in a two-dimensional layout}

\author{Tomas Jochym-O'Connor}
\affiliation{
    Institute for Quantum Computing and Department of Physics and Astronomy,
    University of Waterloo,
    Waterloo, Ontario, N2L 3G1, Canada
    }
\author{Stephen~D. Bartlett}
\affiliation{
	Centre for Engineered Quantum Systems, School of Physics, The University of Sydney, Sydney, NSW 2006, Australia
	}

\begin{abstract}
We introduce a class of three-dimensional color codes, which we call \emph{stacked codes}, together with a fault-tolerant transformation that will map logical qubits encoded in two-dimensional~(2D) color codes into stacked codes and back. The stacked code allows for the transversal implementation of a non-Clifford $\pi/8$~logical gate, which when combined with the logical Clifford gates that are transversal in the 2D color code give a gate set that is both fault-tolerant and universal without requiring nonstabilizer magic states. We then show that the layers forming the stacked code can be unfolded and arranged in a 2D layout.  As only Clifford gates can be implemented transversally for 2D topological stabilizer codes, a nonlocal operation must be incorporated in order to allow for this transversal application of a non-Clifford gate.  Our code achieves this operation through the transformation from a 2D color code to the unfolded stacked code induced by measuring only geometrically local stabilizers and gauge operators within the bulk of 2D color codes together with a nonlocal operator that has support on a one-dimensional boundary between such 2D codes. We believe that this proposed method to implement the nonlocal operation is a realistic one for 2D stabilizer layouts and would be beneficial in avoiding the large overheads caused by magic state distillation.
\end{abstract}


\maketitle

\section{Introduction}

Quantum error correction is a necessary tool for the suppression of logical error rates, enabling sufficiently long coherence times for logical computations. Among the most promising quantum coding architectures are two-dimensional~(2D) local topological stabilizer codes. These are stabilizer codes where each stabilizer measurement couples qubits that are geometrically local on a 2D lattice.  Such schemes are favored due to their relative experimental simplicity of arranging and measuring local stabilizers, typical high error threshold rates, and the ability to vary the distance of the code in a smooth manner rather than through jumps as in concatenated coding schemes. 

To perform universal quantum computation, a fault-tolerant architecture must specify not only a quantum code but also a means to implement a universal set of quantum logic gates.  The most desirable form of logical operation is a transversal gate, that is, a gate where each physical qubit of the code is transformed independently and identically, ensuring that there is no coupling between the different qubits in the code and thereby restricting the propagation of errors.  Unfortunately, there are no quantum codes that allow for the implementation of a universal logical gate set using only transversal gates, as shown by Eastin and Knill~\cite{EK09}.  

The set of transversal gates is even more restricted when considering 2D topological stabilizer codes: Only Clifford gates (a non-universal and classically efficiently simulatable gate set) can be implemented transversally, as originally shown by Bravyi and K\"onig~\cite{BK13} for 2D topological stabilizer codes and subsequently generalized to 2D topological subsystem codes~\cite{PY15}. As an example, 2D color codes are local topological stabilizer codes that have many interesting properties, including transversal logical Hadamard and phase gates~\cite{BM07}, a distinct advantage not shared by the 2D toric code. Additionally, they possess a transversal Controlled-NOT~(CNOT) gate as they are in the CSS code family and as such can implement any Clifford gate transversally.  Unfortunately, due to the restrictions described above, they do not possess a transversal logic gate outside of the Clifford group. Traditional techniques to bypass this problem and obtaining a fault-tolerant non-Clifford gate rely on preparing a special ancillary state~\cite{BK05}, called a magic state, which can lead to large ancilla qubit overhead~\cite{FMMC12}.  

A recent avenue of research for addressing these limitations is to consider the interplay between 2D and three-dimensional~(3D) topological stabilizer codes.  The basis of this approach is a technique for sidestepping the Eastin--Knill no-go theorem through the use of gauge operator measurements to transform from one stabilizer code, with its set of transversal logic gates, to another stabilizer code with a different set of transversal gates~\cite{PR13}; see also~\cite{ADP14}. Applied to topological stabilizer codes, one approach involves mapping a 2D color code to a 3D color code by performing an appropriate set of gauge stabilizer measurements between the 2D code and a specially prepared 3D code ancilla state~\cite{PR13, Bombin14}. The mapping of the quantum information into a 3D color code allows for the application of a transversal $\pi/8$ gate~\cite{Bombin15} (often referred to as the $T$~gate), which is a non-Clifford gate, thus completing the universal gate set. A drawback of such a method is that the required operations are geometrically local only in three dimensions, which may be incompatible with some experimental approaches.  

In this paper, we present a method for fault-tolerantly performing a universal set of quantum logic gates within a 2D architecture.  Our method translates between error correcting codes---a 2D color code, and a special class of 3D color code---to allow for the transversal application of different sets of logical gates.  Specifically, we present a mapping from 2D color codes to a 3D code, which we call the \emph{stacked code}, by pairing multiple copies of the 2D color code, generalizing the work of Ref.~\cite{ADP14}.  Multiple 2D color codes can be pairwise stacked in this manner to increase the overall distance of the newly created stacked code to equal the distance of the 2D color code.  We show that the stacked code admits a transversal $\pi/8$ logic gate, and that the transformation from the 2D color code to the stacked code and back can be performed fault-tolerantly.  Furthermore, by unfolding the stacked code and tiling the original 2D color codes in a 2D layout, this code maintains its properties.   The transformation from 2D color code to stacked code in this 2D layout can be performed with a sequence of local gauge measurements in the bulk of the 2D color codes and Bell pairing measurements along the boundary of neighboring 2D codes. In order to not violate the Bravyi--K\"onig no-go theorem, the measurements pairing the different 2D color codes are necessarily nonlocal, but in a very limited way.  Specifically, these measurements can be performed along one-dimensional~(1D) strips forming the boundary between neighboring 2D codes in a 2D arrangement.

A recent result by Bravyi and Cross~\cite{BC15} presents a very similar construction to the one we present here.  Specifically, they detail a fault-tolerant 2D construction for universal quantum computation that relies on the same type of pairing of 2D color codes (which they call doubled color codes) and measurements between the different layers of color codes, as we propose, to implement a gate outside the Clifford group. Importantly, the results of Ref.~\cite{BC15} go beyond our construction to show how to reduce the nonlocal joint logical Pauli operators that have to be measured into a sequence of local measurements by proposing a ``subdivision gadget." They further supplement their findings by proposing a decoding method to address for the correlated noise that is introduced by the action of the non-Clifford $\pi/8$~gate. In addition, Jones, Brooks, and Harrington recently proposed a method to implement a similar form of construction for the $[4.8.8]$~color code~\cite{JBH15}, as opposed to the hexagonal color code studied here and in Ref.~\cite{BC15}. In their construction, they propose a method for measuring the set of nonlocal joint logical Pauli operators through a series a local measurements inspired by lattice surgery methods~\cite{HFDvM12, LR14}. Our results complement those of Refs.~\cite{BC15, JBH15} by providing an explicit presentation of the properties of this 2D structure as a type of 3D color code with stabilizers that can be inferred by measurements only of local 2D stabilizers and gauge operators, as well as weight-$O(d)$ 1D operators on the boundaries of 2D codes.

Our paper is structured as follows. In Sec.~\ref{sec:TransformingCodes}, we introduce the stacked code and demonstrate how to transform from the 2D color codes to this stacked code. We prove minimum error distance and other properties of the stacked code, and we show how to implement the non-Clifford logical gate fault-tolerantly by transformating from a 2D color code to the stacked code and back.  In Sec.~\ref{sec:quasi2D}, we show how to unfold the stacked code into a 2D layout, and present a method to implement the transversal $\pi/8$ logic gate by an appropriate set of local 2D gauge measurements and nonlocal 1D strip measurements. We also compare our scheme to the recent result of Bravyi and Cross.  Finally, in Sec.~\ref{sec:Distances}, we present a theoretical argument measuring the degree of nonlocality of our operations with respect to a higher-distance 2D code.  Some brief concluding remarks are given in Sec.~\ref{sec:Conclusion}, and details on transversal gate operations in the stacked code are relegated to the appendix.

\section{Transforming to the stacked code}
\label{sec:TransformingCodes}

In this section, we describe a transformation to map the logical qubit encoded in a 2D color code into a particular form of 3D color code, which we call a \emph{stacked code}.  This stacked code will allow for the transversal implementation of a logical $\pi/8$~gate (defined by $\text{diag}[e^{-i \pi/8}, e^{i \pi/8}]$), which together with the transversal logical gates in the 2D color code form a universal gate set. We introduce this transformation by generalizing the technique of Anderson~\textit{et al.}~\cite{ADP14}, which mapped a seven-qubit Steane code (also a $d=3$ 2D color code) to a 15-qubit quantum Reed-Muller code (also a $d=3$ 3D color code).   Our generalization applies to hexagonal color codes of any distance, and gives rise to a 3D color code of distance $d=3$. We then show to further generalize this transformation to yield a stacked code with arbitrary distance~$d$.

\subsection{Transforming 2D color codes to 3D: distance~3 protection}

Consider a $[[n,1,d]]$ hexagonal color code family~\cite{BM07b}, with $n = (3d^2 + 1)/4$, defined by $X$ and~$Z$ stabilizer generators expressed as plaquette operators $G_{P_i} = \otimes_{\nu \in P_i} X_\nu$ and $H_{P_i} = \otimes_{\nu \in P_i} Z_\nu$, where the tensor product is over vertices $\nu$ defining a hexagonal plaquette $P_i$, with appropriate modification at the boundaries.  Our construction will use multiple copies of such codes with stabilizer generators $\{ G_{P_i}^{(l)} \}$ and $\{ H_{P_i}^{(l)} \}$, where $l$~is a label for the particular copy of the 2D color code.  For any such code, one can identify a set of weight-$2$ $Z$-type edge operators~$\{H_{e_i}^{(l)} \}$, see Fig.~\ref{fig:GaugeEdge}, that will, along with the $Z$-type plaquette operators, generate any $Z$-type edge in the 2D lattice. We label these edges by~$e_i$, as they can be identified in a one-to-one correspondence with plaquette operators labeled by~$P_i$. Given such a generating set~$\{H_{e_i}^{(l)} \}$, one can identify each $X$~plaquette generator~$ G_{P_i}^{(l)} $ with a particular $Z$~edge operator~$H_{e_i}^{(l)}$ such that this pair of operators  anti-commute, as they will intersect at only one site. 

\begin{figure}
\centering
\begin{subfigure}{0.2\textwidth}
\includegraphics[width=0.5\textwidth]{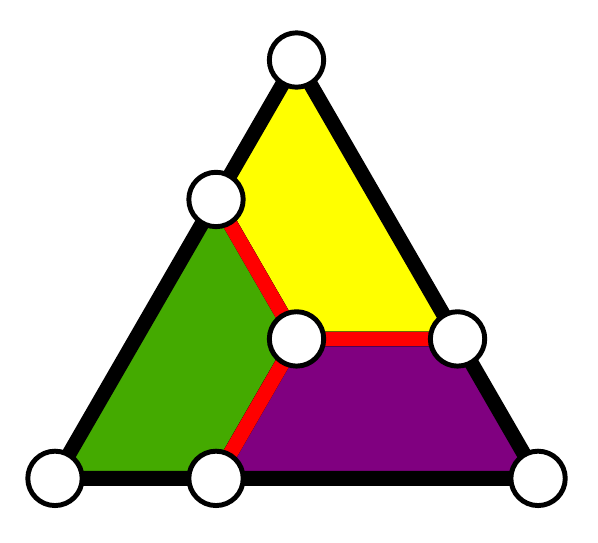}
\caption{}
\label{fig:2D_d3}
\end{subfigure}
\begin{subfigure}{0.25\textwidth}
\includegraphics[width=0.8\textwidth]{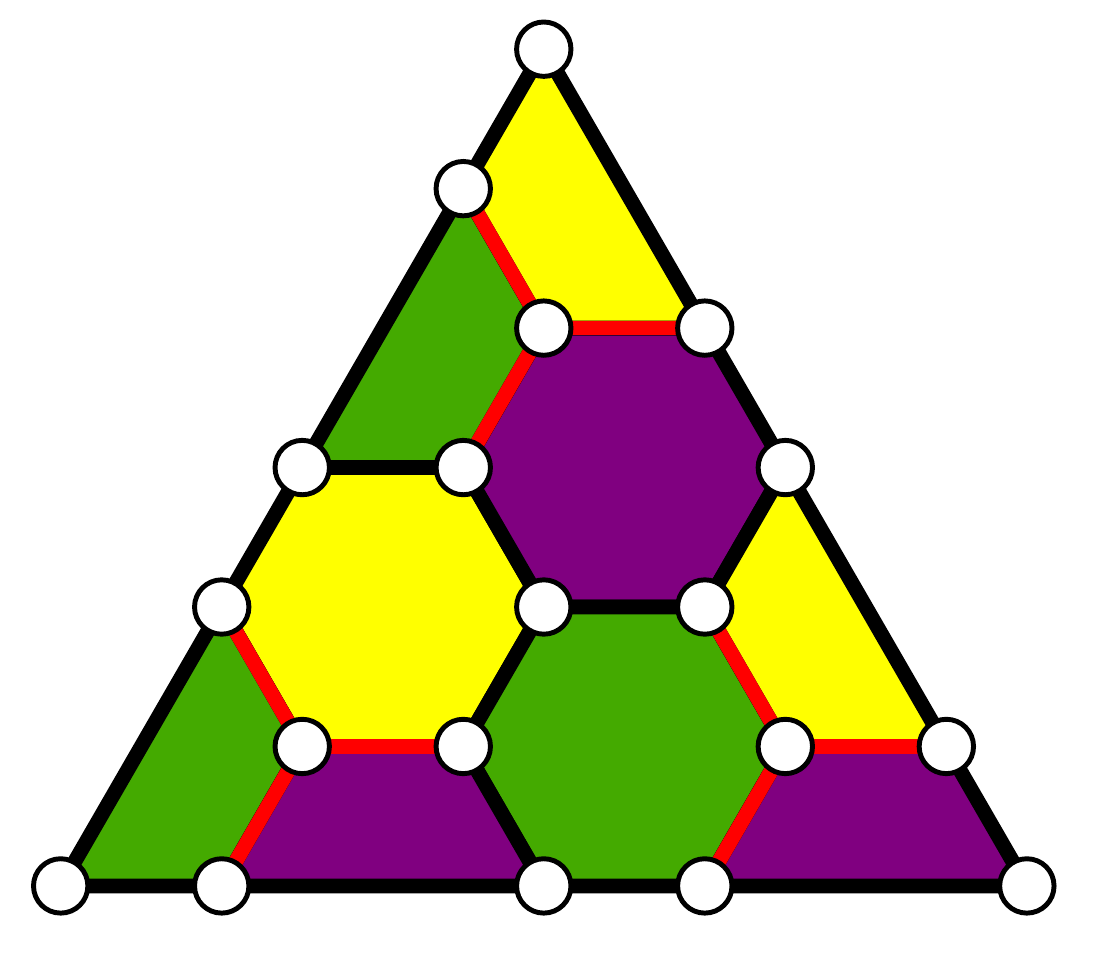}
\caption{}
\label{fig:2D_d5}
\end{subfigure}
\caption{Two instances of the 2D hexagonal color code of distance~(\subref{fig:2D_d3})~$d=3$ and~(\subref{fig:2D_d5})~$d=5$. In each case, a set of independent edges~$\{ H_{e_i} \}$, shown in red, can be chosen as the set that will form the $Z$~gauge operators when paired with the identical edge from another code copy, thus forming weight-4 gauge operators~$\{ H_{e_i}^{(2k-1)} H_{e_i}^{(2k)} \}$.}
\label{fig:GaugeEdge}
\end{figure}

Consider a logical qubit encoded in a 2D hexagonal color code labeled $l=1$ of distance~$d$, with stabilizer generators $\{ G_{P_i}^{(1)} \}$ and~$\{ H_{P_i}^{(1)} \}$.  We now consider a process by which we transform this 2D code into a 3D code, following the method of Anderson~\textit{et al.}~\cite{ADP14}.  Our transformation makes use of a second 2D color code of equivalent size to the first, with its encoded logical qubit entangled in a Bell state with a single ancilla qubit. That is, denoting the logical operators for the second color code labeled $l=2$ by $X^{(2)}_L$ and $Z^{(2)}_L$, and the operators for a single ancilla qubit by $X$ and $Z$, this Bell state is stabilized by $X_L^{(2)} \otimes X$ and $Z_L^{(2)} \otimes Z$ as well as the code stabilizers $\{ G_{P_i}^{(2)} \}$ and~$\{ H_{P_i}^{(2)} \}$.

We induce the transformation through joint measurement of gauge operators of the two color codes.  Specifically, we measure the $Z$-type gauge operators between the two codes corresponding to pairing up the generating $Z$-type edge operators of the two codes and jointly measuring the corresponding weight-4 operators~$\{ H_{e_i}^{(1)}  \otimes  H_{e_i}^{(2)} \}$.  Because each of the original $X$-type plaquette operators of the two codes~$G_{P_i}^{(l)}$ anti-commute with the measured gauge operator~$H_{e_i}^{(1)}  \otimes  H_{e_i}^{(2)}$, they will no longer be stabilizers of the code. However, the joint volume operator~$G_{P_i}^{(1)}  \otimes  G_{P_i}^{(2)}$ obtained by pairing corresponding plaquette operators between the two code copies will remain a stabilizer as it has even overlap with the gauge operator~$H_{e_i}^{(1)}  \otimes  H_{e_i}^{(2)}$. As a result of these measurements, the evolution of the stabilizers for the entire system is given by:  

\begin{tabular}{p{4cm}||p{4cm}}
\centering 2D code + ancilla Bell
{\begin{align}
& \{ G_{P_i}^{(1)} \} \otimes I^{\otimes n} \otimes I  \nonumber \\    
&\{ H_{P_i}^{(1)} \} \otimes I^{\otimes n} \otimes I \nonumber \\
&\{ G_{P_i}^{(1)}  \otimes  G_{P_i}^{(2)} \} \otimes I \nonumber \\
&\{ H_{P_i}^{(1)} \otimes  H_{P_i}^{(2)} \} \otimes I \nonumber \\
&I^{\otimes n} \otimes X_L^{(2)} \otimes X \nonumber \\
&I^{\otimes n} \otimes Z_L^{(2)} \otimes Z \nonumber 
\end{align}}
&
\centering 3D code
{\begin{align}
&\{ H_{e_i}^{(1)}  \otimes  H_{e_i}^{(2)} \} \otimes I \\    
&\{ H_{P_i}^{(1)} \} \otimes I^{\otimes n} \otimes I \\
&\{ G_{P_i}^{(1)}  \otimes  G_{P_i}^{(2)} \} \otimes I \label{eq:Xcube}\\
&\{ H_{P_i}^{(1)} \otimes  H_{P_i}^{(2)} \} \otimes I  \label{eq:Zcube} \\
&I^{\otimes n} \otimes X_L^{(2)} \otimes X  \label{eq:XBell} \\
&I^{\otimes n} \otimes Z_L^{(2)} \otimes Z \label{eq:ZBell}
\end{align}  } \\ \end{tabular} 
where the last two stabilizers represent those corresponding to the second code copy being prepared in a Bell pair with an ancilla qubit.  We note that choosing the smallest nontrivial 2D color code, corresponding to $n=7$ and $d=3$ and equivalent to the seven-qubit Steane code, this mapping corresponds to that of Anderson~\textit{et al.}~\cite{ADP14} in this case.  Even though in general the 2D codes used in this construction are of distance~$d$, the overall distance of transformed code is limited to be~3. Logical~$Z$ string operators are formed by matching pairs of qubits from the two copies of the 2D codes along with the single ancilla qubit, and take the form~$Z_i^{(1)} Z_i^{(2)} Z$.  A higher weight logical~$Z$ operator can be obtained by traversing the 2D color code layers and connecting error strings of different colors. We shall expand upon this point for the general case in Sec.~\ref{sec:3Dproperties}.

This new code is a 3D color code, where the 3D code stabilizers of Eqs.~\eqref{eq:XBell}--\eqref{eq:ZBell} correspond to the stabilizers of the fourth color and the boundary of the new color corresponds to the original 2D code. We prove that it is a 3D color code, and determine its distance in the general case, in Sec.~\ref{sec:3Dproperties}.  The code possesses a transversal~$\pi/8$~gate, as proven in Appendix~\ref{app:GateTransversality} in a similar manner to the techniques proposed in Refs.~\cite{BM07, BH12, KB15}, and will therefore form a universal fault-tolerant gate set along with the logical Clifford gates that can be applied transversally to the original 2D code.\footnote{The transversal gates are not strictly transversal, that is all the same rotation, for the hexagonal color code. However, by applying the inverse rotation to the appropriate set of qubits the correct logical operator can be applied~\cite{Bombin15, KB15}.}  

However, this code has a number of undesirable features from the perspective of topological stabilizer codes.  First, we note that the stabilizers in Eqs.~(\ref{eq:XBell})-(\ref{eq:ZBell}) are very high weight, having support on the entire set of qubits across a full 2D layer.  We postpone discussion about how one might infer the values of these high-weight stabilizers using only lower-weight measurements to Sec.~\ref{sec:FTPi8}.  Second, the distance of this 3D code is limited by the width of the third dimension (two~layers + one~ancilla qubit). This limitation is in line with the intuition behind the no-go result of Bravyi and K\"onig~\cite{BK13}, where it is shown that a topological stabilizer code must be at least dimension 3 or higher to possess a transversal gate operation that lies outside the Clifford group.  One might suspect that the fault-tolerance protection that one should get from the distance of the code should be related to the depth of the third dimension of the code.

\subsection{Transforming 2D color codes to 3D: distance~$d$ protection}

To increase the distance of our newly formed code, we must increase the width of its third dimension. A natural method to provide such added protection would be to encode the weakest part of the code, the bare ancilla qubit, into a 3D code of its own using the exact same technique. We can continue this process  recursively, by performing the joint stabilizer measurements in~\eqref{eq:XBell}--\eqref{eq:ZBell} as joint logical~$X$ and $Z$~measurements. The encoded ancilla state will be prepared offline using 2D color codes arranged as layers in a stack, coupled into logical Bell pairs by performing joint logical~$X$ and $Z$~measurements, henceforth referred to as \emph{Bell stabilizers}. This bulk ancilla state will allow us to transform our 2D color code into a 3D color code with large distance.  In addition, as the individual components forming the bulk ancilla state are restricted to pairs of 2D layers, this will allow us to show in Sec.~\ref{sec:quasi2D} that such a process can be made fault-tolerant on a 2D lattice.

Specifically, our recursive transformation from a 2D color code on layer $k=1$ to a $d$-layer stack is defined by the following evolution of stabilizers: 

\begin{tabular}{p{4cm}||p{4cm}}
\centering 2D code + ancilla Bell 
{\begin{align}
&\{ G_{P_i}^{(2k-1)} \} \nonumber \\    
&\{ H_{P_i}^{(2k-1)} \} \nonumber \\
&\{ G_{P_i}^{(2k-1)}  G_{P_i}^{(2k)} \} \nonumber \\
&\{ H_{P_i}^{(2k-1)}  H_{P_i}^{(2k)} \} \nonumber \\
& X_L^{(2k)} X_L^{(2k+1)} \nonumber \\
& Z_L^{(2k)} Z_L^{(2k+1)} \nonumber 
\end{align}}
&
\centering 3D code
{\begin{align}
& \{ H_{e_i}^{(2k-1)}  H_{e_i}^{(2k)} \} \\    
&\{ H_{P_i}^{(2k-1)} \} \\
&\{ G_{P_i}^{(2k-1)}  G_{P_i}^{(2k)} \} \label{eq:XcubeGen}\\
&\{ H_{P_i}^{(2k-1)} H_{P_i}^{(2k)} \} \label{eq:ZcubeGen} \\
& X_L^{(2k)} X_L^{(2k+1)}  \label{eq:XLBell} \\
& Z_L^{(2k)} Z_L^{(2k+1)} \label{eq:ZLBell}
\end{align} } \\ \end{tabular} 
where $k \in \{1, \cdots, \frac{d-1}{2} \}$.  As the final layer is a single qubit, we have $X_L^{(d)} = X$, and $Z_L^{(d)}=Z$. The logical qubit is initially stored in the first 2D~color code layer, stabilized by the operators~$\{ G_{P_i}^{(1)} \}$~and~$\{ H_{P_i}^{(1)} \}$. The additional layers are prepared in joint Bell pairs, as indicated by the Bell stabilizers~$X_L^{(2k)} X_L^{(2k+1)}$ and~$Z_L^{(2k)} Z_L^{(2k+1)}$. The pairs of copies of the 2D sheets are then coupled together by measuring the gauge operators~$\{ H_{e_i}^{(2k-1)}   H_{e_i}^{(2k)} \}$ between one sheet and another sheet from a different pair. This is logically equivalent to stacking the different pairs to form one large stack of height distance~$d$, where each layer is a copy of a 2D color code also with distance~$d$, as shown in Fig.~\ref{fig:3Dstack}.  We call the resulting 3D code the \emph{$(d-1)+1$ stacked code}. At this point, the Bell stabilizers in Eqns.~\eqref{eq:XLBell}--\eqref{eq:ZLBell} have a cell-like structure connecting the two 2D color code sheets with which they are associated. These correspond to the Blue stabilizers in Fig.~\ref{fig:3Dstack} and will have particular features when viewing this code as a 3D color code, as we explore in the next section, as well as several properties needed to make our 2D arrangement of this code in Sec.~\ref{sec:2Dlayout}.

\begin{figure*}
\begin{center}
\begin{subfigure}{0.4\textwidth}
\includegraphics[width = \textwidth]{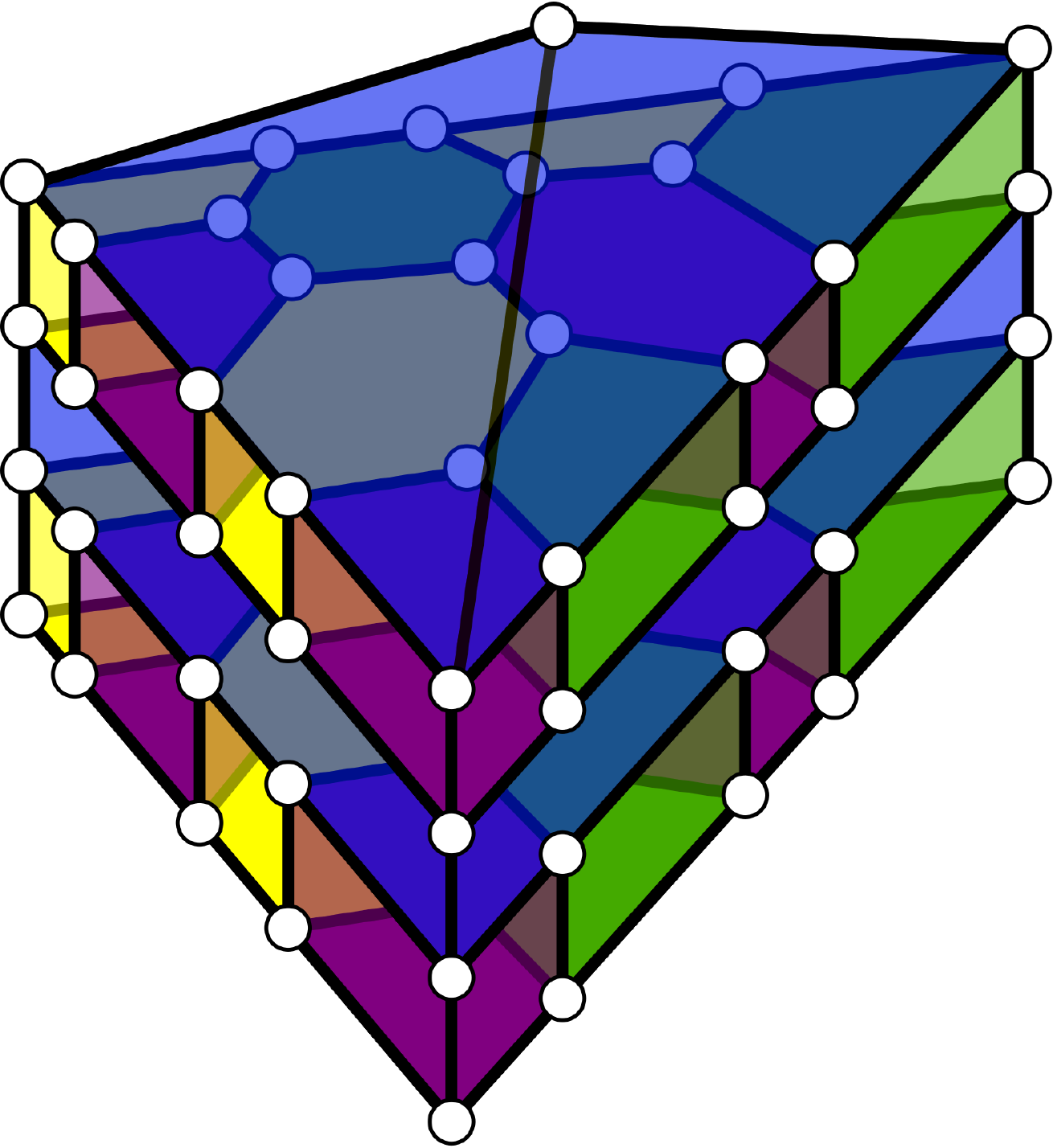}
\caption{}
\label{fig:Primal_3D}
\end{subfigure}
\begin{subfigure}{0.4\textwidth}
\includegraphics[width = 0.66\textwidth]{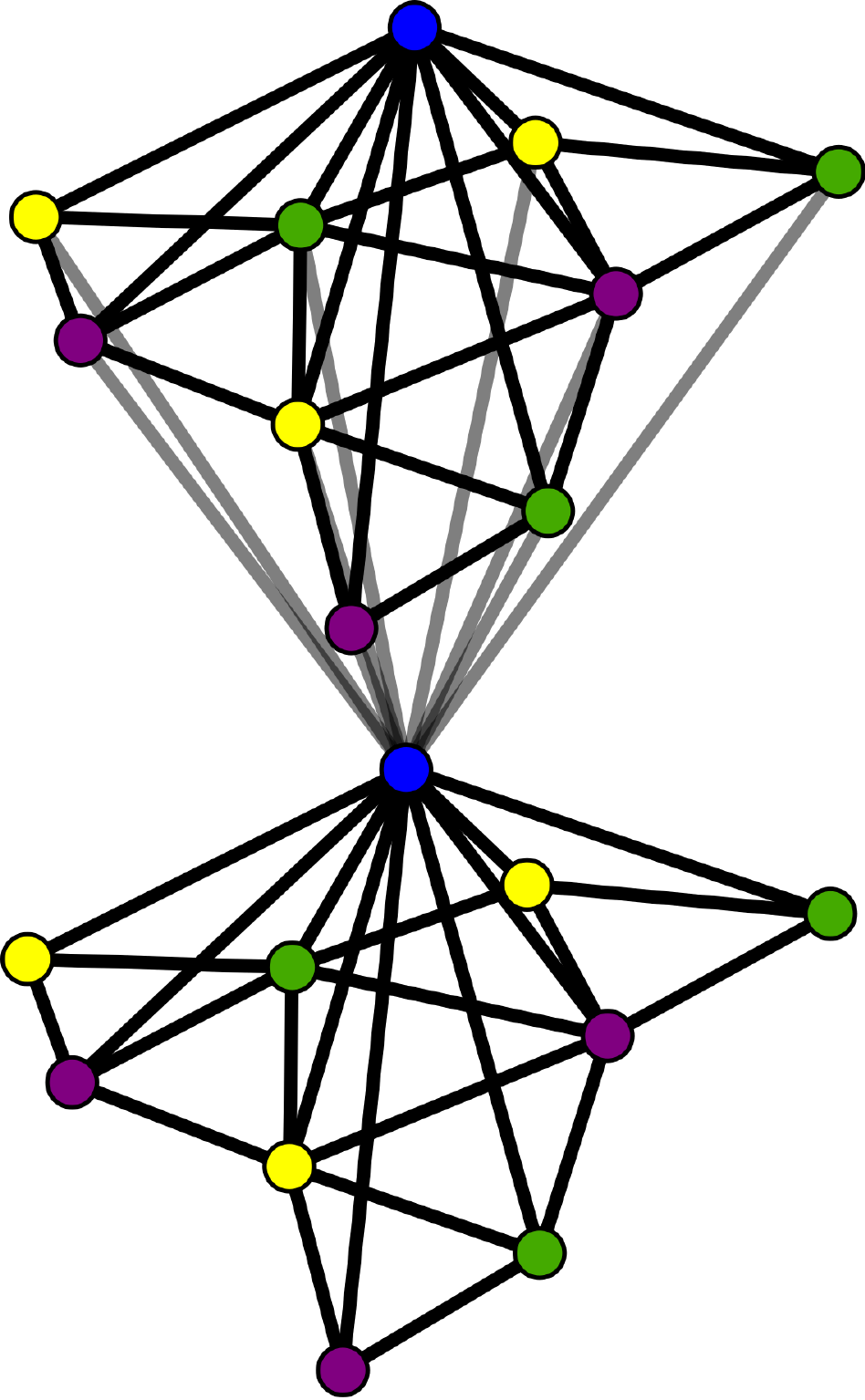}
\caption{}
\label{fig:Dual_3D}
\end{subfigure}
\caption{(\subref{fig:Primal_3D}) Graphical representation of the primal lattice of the $(d-1)+1$~stacked code formed by stacking different copies of 2D color codes, shown here for $d=5$. The copies of the 2D code are coupled either by measuring gauge operators or logical operator pairs (shown in blue) between the different layers. (\subref{fig:Dual_3D}) Dual lattice for the 3D stacked code ($d=5$). Vertices represent cell stabilizers in the primal lattice and edges represent faces shared by connected stabilizers.}
\label{fig:3Dstack}
\end{center}
\end{figure*}

\subsection{Properties of the stacked code}
\label{sec:3Dproperties}

The $(d-1)+1$ stacked code is also a 3D color code.  This can most easily be seen using its dual lattice, as follows. Take the dual lattice of the 2D color code, connect each of the vertices of the dual lattice (consisting of 3 colors) to a single vertex of a different color. We shall denote the colors of the original 2D code as green~$(g)$, purple~$(p)$, and yellow~$(y)$ and the color of the newly formed stabilizers in 3D by blue~$(b)$. Connect this single vertex to another set of vertices forming a 2D code, and repeat this process $(d-1)/2$~times. Each of the vertices in the dual lattice form a 3D stabilizer cell in the primal lattice, where edges between the vertices in the dual lattice are equivalent to faces at the intersection of cells in the primal lattice, see Fig.~\ref{fig:3Dstack} for an example of the dual lattice. It is  straightforward to see that this construction is equivalent to the construction outlined for the stacked code, and moreover, because the dual lattice is four-colorable and composed of tetrahedra, it is a valid 3D color code~\cite{BM07, KB15}.

We now proceed to determine the distance of the $(d-1)+1$ stacked code, making use of the well-studied properties of the 3D color code. The edges in the primal lattice of a color code can be identified with one of the colors of the code~\cite{Bombin15, KB15}. In the case of a 3D color code, the faces at the intersection of two tetrahedra in the dual lattice correspond to edges in the primal lattice, where the color of the edge in the primal lattice is given by the complementary color to the vertices forming the face in the dual lattice. A boundary of a given color is the set of points at which edges of that given color terminate without a stabilizer of the given color being present. In the case of the stacked code, the three original colors of the 2D lattice form boundaries along the three sides of the stack extending upwards from their original 1D boundary given by the 2D color code. The fourth boundary, for the newly introduced color in three dimensions (blue), is located along the bottom boundary of the 3D lattice, as none of these qubits touch a blue stabilizer.

The Bell stabilizers given in Eqs.~\eqref{eq:XLBell}--\eqref{eq:ZLBell} correspond to the blue stabilizers in Fig.~\ref{fig:3Dstack}, and are equivalent to measuring the joint logical~$X$ and~$Z$ operators of the two 2D color codes forming the top and bottom faces of the blue stabilizer. As opposed to traditional constructions of 3D color codes, the Blue stabilizers are not of low weight, but rather act on~$\mathcal{O}(d^2)$ qubits. This is a particular feature of the stacked code structure, as the Blue stabilizers measure joint logical operators across pairs of 2D sheets and thus must contain all qubits across those faces. However, as we show in Sec.~\ref{sec:FTPi8}, these high weight stabilizers across the full 2D sheets need not be measured in practice.

\begin{figure*}
\centering
\begin{subfigure}[t]{0.3\textwidth}
\includegraphics[width=\textwidth]{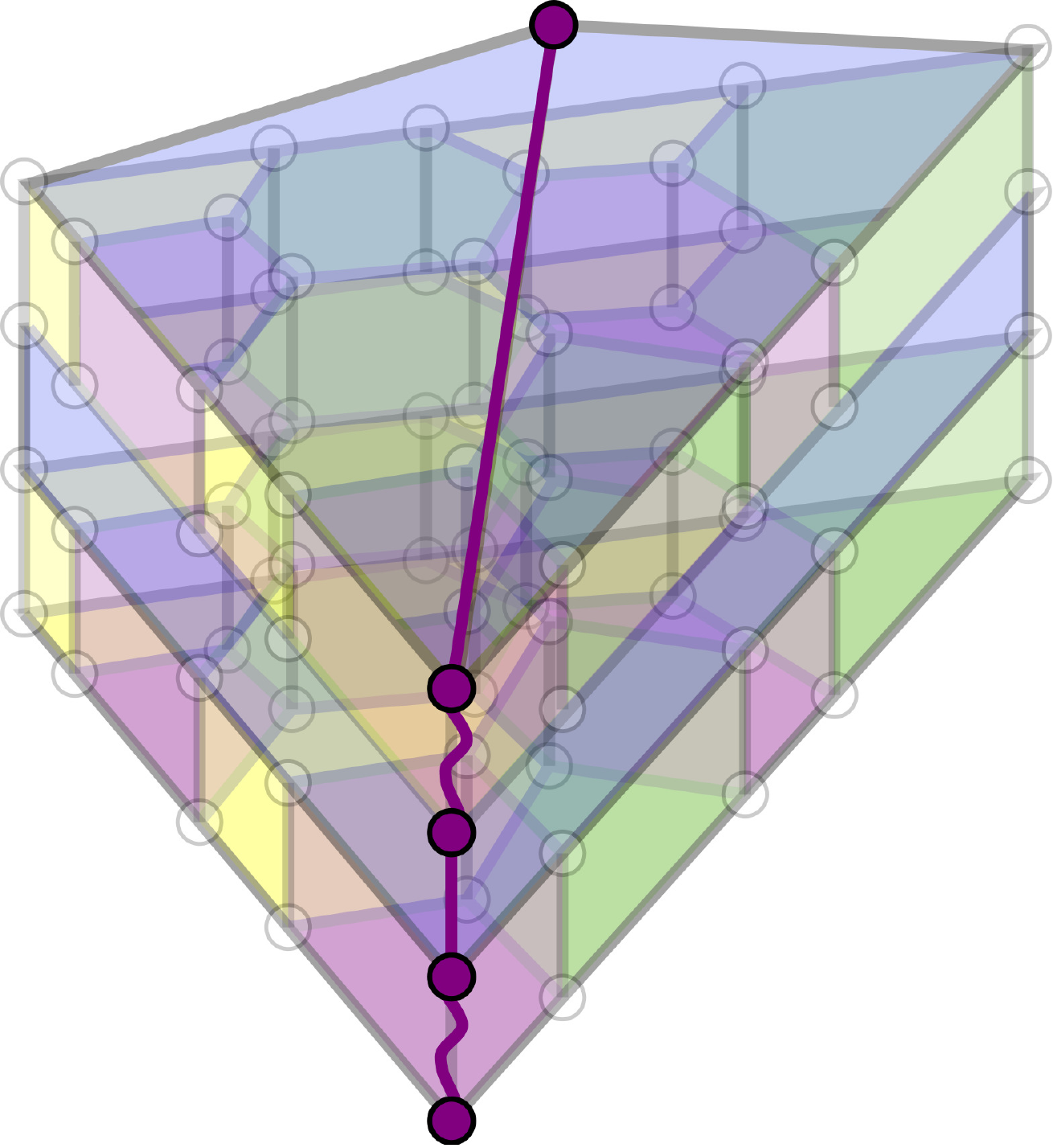}
\caption{}
\label{fig:3DerrorColor}
\end{subfigure}
\begin{subfigure}[t]{0.3\textwidth}
\includegraphics[width=\textwidth]{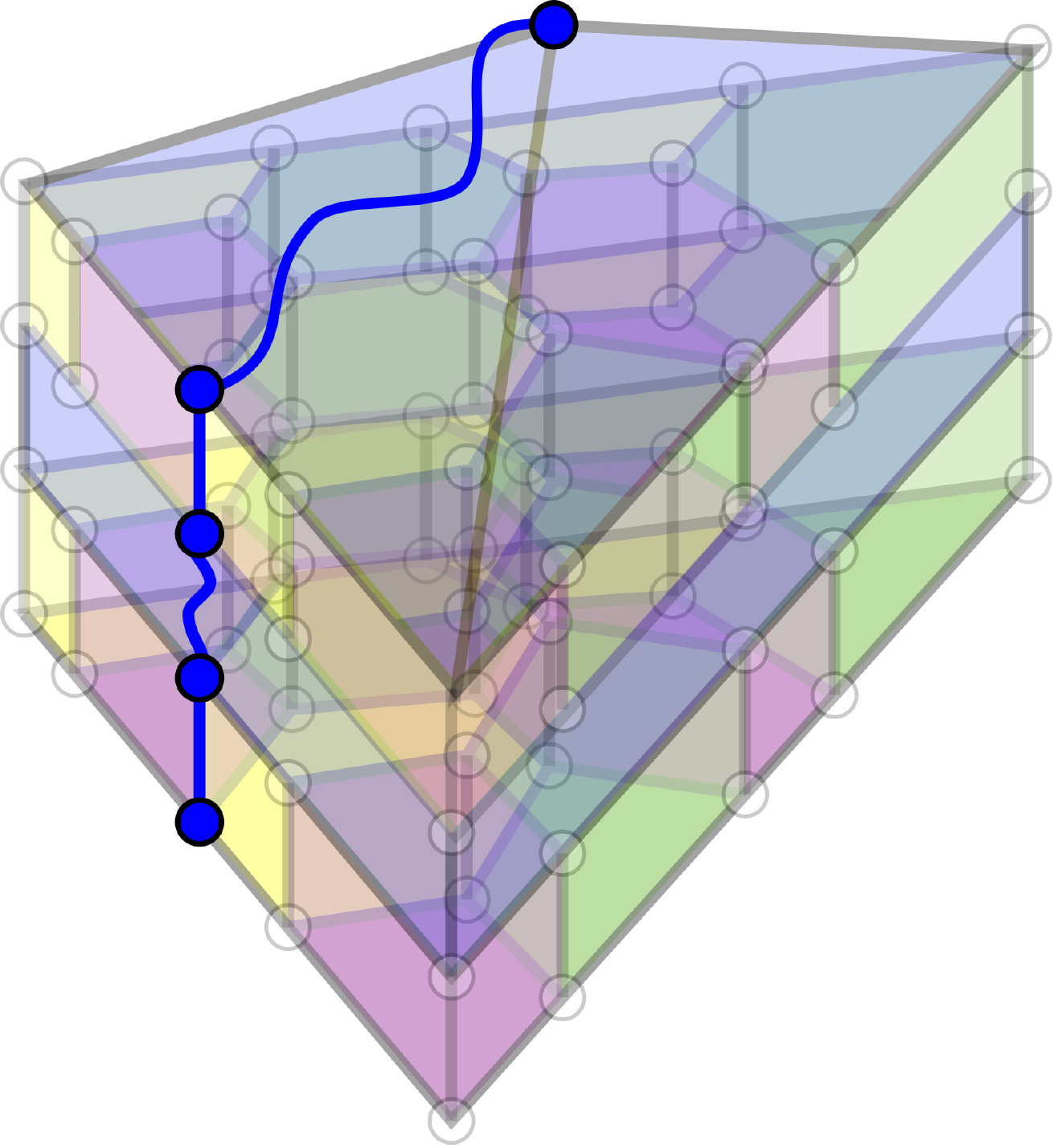}
\caption{}
\label{fig:3DerrorBlue}
\end{subfigure}
\begin{subfigure}[t]{0.3\textwidth}
\includegraphics[width=\textwidth]{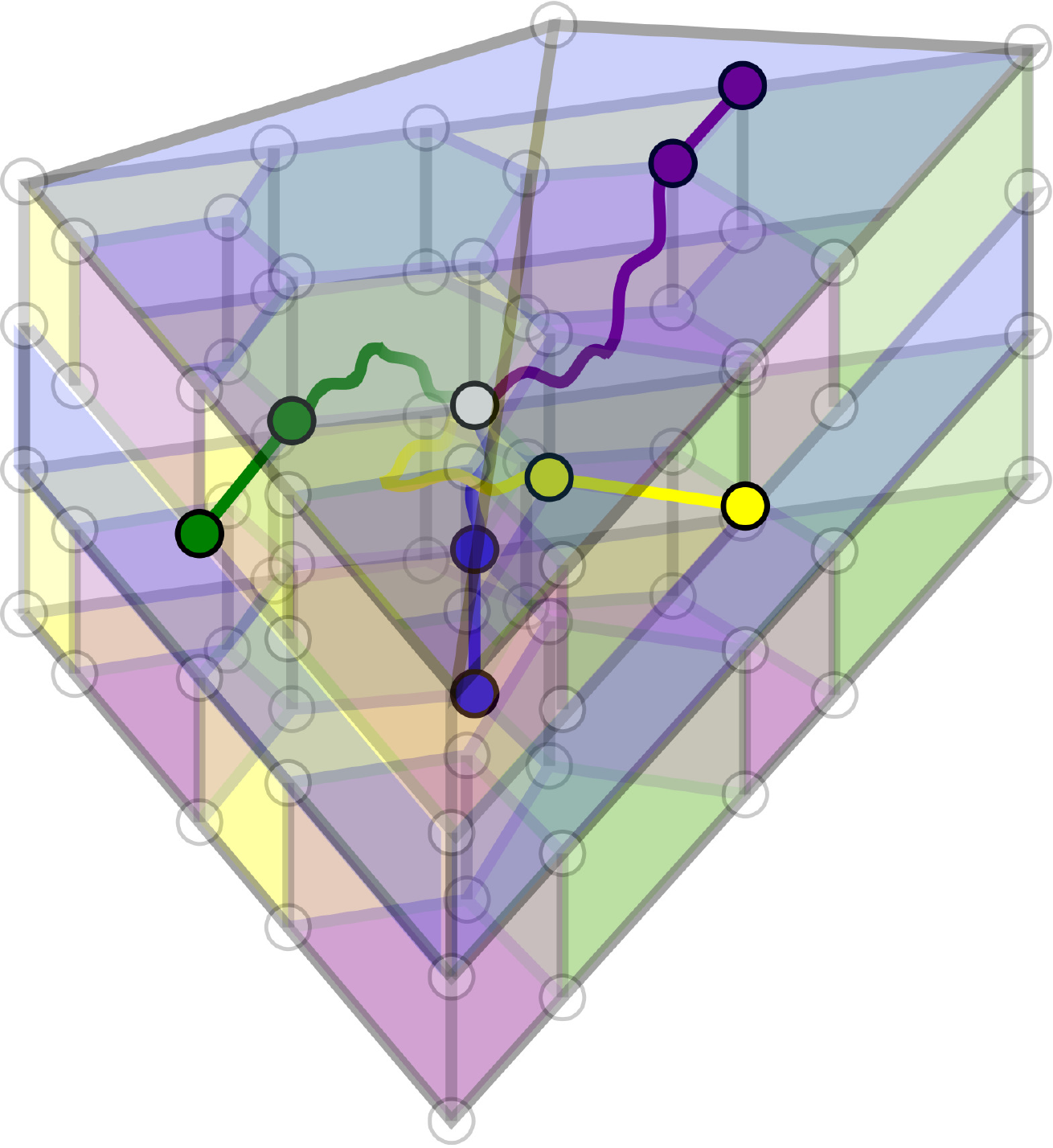}
\caption{}
\label{fig:3DerrorMulti}
\end{subfigure}
\caption{Examples of the different representations of equivalent logical error strings that exist in the 3D stacked code. The color of the logical strings are chosen according to the color of the edges they follow. The curved lines represent joining of edges through a stabilizer of the same color. In (\subref{fig:3DerrorColor}), because the string lies on the green--yellow boundary, it can be chosen to be either of the complementary colors, blue or purple. In (\subref{fig:3DerrorBlue}), the error string connects the bottom blue boundary to the joint boundary of the other three colors at the ancilla qubit, following blue edges. Example~(\subref{fig:3DerrorMulti}) shows how multiple colored boundaries can fuse in the bulk, thus negating the excitation that would otherwise be present.}
\label{fig:3Derrors}
\end{figure*}

Logical operators in any color code are given by string operators that connect the boundaries of different colors~\cite{BM07}.  A $c$-colored string operator is given by a set of qubits formed of connected edges of color~$c$ (two edges are connected if they share a stabilizer of color~$c$). A $c$-colored string operator either has endpoints at the boundary of color~$c$, in which case the final edge of this string connects the endpoint to the boundary, or in the bulk where the endpoint is located at a particular $c$-colored stabilizer, thus causing an excitation. If all of the colored strings meet at a given qubit, then the strings can ``fuse" and the bulk excitation formed by this endpoint will be negated~\cite{BM07, BM07b}. Therefore, in order to obtain a logical string operator, all colored string operators must connect their respective boundary to a shared fusion point, leading to a nontrivial string connecting boundaries of all colors without excitations. These properties now allow us to prove the distance of the stacked code.

\begin{lemma}
\label{lem:3DstackDist}
A $(d-1)+1$ stacked code is a 3D color code whose distance is~$d$.
\end{lemma}

\begin{proof}
The stacked code comprises pairs of 2D layers separated by large blue~$(b)$ stabilizers. We shall consider two different representations of logical $Z$~operators, one where the logical operator is composed of qubits that are only in a single pair of these 2D layers, and one where the logical operators span multiple pairs of 2D layers. In the first case, the only way for such a logical operator to connect to the $b$~boundary would be for it to be in the bottom-most pair of 2D layers, as they themselves are trivially connected to the $b$~boundary. However, because we are focusing on a single pair of 2D~layers, we can map the problem of finding a logical operator to that of finding one in a single 2D layer, where connecting edges correspond to one of the pair of edges connecting two stabilizer cells of the same color (these edges correspond to the original edges of the individual 2D codes). If the same edge in both color code copies is part of the error chain, then these two edges cancel out as the resulting face corresponds to a gauge~$Z$ operator. Therefore, we can refer back to individual edges connecting stabilizers in the 2D color code. As such, because the 2D code is a code of distance~$d$, the smallest-weight logical string that connects the different colored boundaries must be weight~$d$, and therefore any such logical operator will be of distance~$d$.

Suppose we are given a set of $Z$~errors forming a string operator of one of the colors of the original 2D code.  Without loss of generality, let this string be of color~$g$. Now, given that a string operator formed by a set of edges of color~$g$, the only way for a string operator of color~$g$ to connect qubits from different pairings of 2D layers (that is, traverse a blue stabilizer) is by using $g$-colored edges at the corner of a given layer. These points lie at the joint boundary of $p$~and~$y$ by definition. There are then two methods for such an error string to connect to the $g$~boundary, either by traversing through a given 2D pair to the $g$~boundary of the other side, or by connecting up to the single ancilla qubit that is at the intersection of the~$g$, $p$, and $y$ boundaries. In the case of the former, in order for a logical string to connect across a given pair of 2D layers to the boundary on the other side, the minimal weight will be governed by the distance of the individual 2D codes, as we previously saw. Therefore, the minimal weight of such a logical operator will be~$d$. In the case when the error string connects to the single ancilla qubit, then in order to form a logical operator it must also connect to the $b$~boundary, as shown in Fig.~\ref{fig:3Derrors}\subref{fig:3DerrorColor}. The single ancilla qubit is as far away from the $b$~boundary as it can be, and in order to create a logical string that connects to the bottom boundary through $g$~edges, there will have to be at least a single qubit per 2D layer connecting to the ancilla qubit. Therefore, the minimal distance of such an operator will also be~$d$. Finally, we must consider the case where the logical string is not composed of strings of colors~$\{ g, p, y \}$ (the original colors of the 2D code). In such a case, the string operator must terminate at the joint boundary of the three colors, again given by the single ancilla qubit, and as in the previous case must connect the single qubit to the bottom $b$~boundary through a $b$-colored chain, as shown in Fig.~\ref{fig:3Derrors}\subref{fig:3DerrorBlue}. Such an operator will be of weight at least~$d$ as argued above. As such, the minimum weight non-trivial $Z$~operator is of weight~$d$. Since logical~$X$ operators are formed by connecting 2D membranes in the 3D code~\cite{BM07}, the $X$~logical distance will be greater than that of the $Z$~logical operators, and as such the distance of the code is~$d$.
\end{proof}

We note a potential efficiency that may be gained in the number of qubits in the stacked code.  Because the distance to the blue boundary (the bottom layer) of each pair of 2D code sheets increases by 2 for each separation by a blue stabilizer, as shown in Fig.~\ref{fig:3Dstack}, we can in principle use pairs of 2D color codes of decreasing distance according to how far away they are from the blue boundary, i.e., decreasing with $k$. Although we do not prove this result here, the intuition behind this idea is as follows.  Because a logical error must connect to the blue boundary, there is extra protection for any logical error that wants to span a given pair of 2D sheets as the error string will have to traverse all layers below the pair of layers. The stacked code prepared in such a way would resemble more of a pyramid than a prism.  This method of stacked code construction leads to an analogous code as presented by Bravyi and Cross, based on differing sizes of doubled color codes~\cite{BC15}.

\subsection{Fault-tolerant implementations of a universal gate set}
\label{sec:FTPi8}

Consider a qubit encoded into a 2D hexagonal color code.  By the properties of this code, logical Hadamard~$H$ and Phase~$S$ are transversal, and a logical CNOT between two such codes is also transversal. These are all logical Clifford gates, and so we require an additional gate such as the logical $\pi/8$~gate to complete a universal set.  As we now show, transforming to the 3D stacked code can be used as a means to complete a universal gate set, just as gauge fixing provides a means for dimensional jumps in gauge color codes~\cite{Bombin14, Bombin15}.

The initial ancillary 2D layers can be prepared in their appropriate Bell pairs offline.  Because these states are stabilizer states, they can be prepared fault-tolerantly.  In order to preserve the fault-tolerance property of the high-weight Bell stabilizers measurements, a cat state of the same number of qubits as the weight can be prepared fault-tolerantly offline~\cite{Shor96, DS96}. The measurement of these high-weight stabilizers is repeated in order to ensure fault-tolerance~\cite{AGP06}. Note that this preparation process can be combined with the final measurement process outlined below, and therefore will not contribute to the overall runtime to complete the operation.

With the ancilla layers prepared in the appropriate state, the transformation from the $k=1$ 2D color code to the stacked code can be induced by measuring the gauge operators in a fault-tolerant manner similar to that of surface code, such that any errors do not spread between data qubits. At this point, the logical qubit is stored throughout the different stacks in the $(d-1)+1$ stacked code.  We emphasize that the high-weight stabilizers of the stacked code are \emph{not} measured at this stage.  Rather, the logical transversal $\pi/8$~gate is performed, and we then immediately transform back to the 2D code (without any active error correction being performed on the stacked code).  The transformation back to the 2D color code is induced by measuring the \textit{original} stabilizers of the 2D code, and the ancillary 2D stacks and their Bell stabilizers. Because the measurements can be performed fault-tolerantly without spreading errors, the code is protected by a distance~$d$ code at all times, and any error that occurred throughout the process can be inferred from the final measurements, as explained below. 

Having returned to the original 2D code, the computation can continue with the application of transversal Clifford gates before potentially doing the same process for another~$\pi/8$~gate at a different point in the computation. It is worth noting that the ancilla state is required to be measured fault-tolerantly through repeated measurements in order to correctly infer the errors on the final 2D color code after completion of the gate. Therefore, this ancilla remains ``ready'' at this stage for future non-Clifford computation and does not have to be re-prepared.

What remains to be shown is how an error that occurs while the information is encoded in the stacked code can be inferred from the final 2D code plus ancilla measurements. Suppose an error of weight less than~$d$ occurred while the state is encoded in the stacked code.  Because the logical~$\pi/8$ gate is transversal, errors may transform but will not increase in weight as a result of the logical gate. Therefore, such an error will remain of weight less than~$d$. As such, if one were to measure the stabilizers of the stacked code, one would see a change in the sign of one of the cell stabilizers. Suppose the error anticommutes with cell~$G_{P_i}^{(2k-1)} G_{P_i}^{(2k)}$ (this corresponds to an $Z$~error, a similar argument follows for $X$~errors). The presence of the error can be inferred from the measurement of the original stabilizers of the 2D planes, because the product of the individual outcomes of measurements~$G_{P_i}^{(2k-1)}$ and $G_{P_i}^{(2k)}$ will be equivalent in sign to the measurement of the cell of the stacked code. It should be noted that the sign of the individual measurements will not necessarily be preserved, because the individual stabilizers of the 2D sheets anticommute with the gauge operators.  However, the effect of these sign changes will simply be to set the stabilizer reference frame for subsequent measurements. Finally, if the error anticommutes with a blue stabilizer of the form~$X_L^{(2k)} X_L^{(2k+1)}$, one can still infer the error from the measurement of the individual operators on the sheets and the joint logical measurements along the shared boundary of the sheets. We return to this last point in Sec.~\ref{sec:quasi2D}.

\section{Unfolding the stacked code: A 2D implementation}
\label{sec:quasi2D}

Our stacked code provides a mechanism for performing a fault-tolerant logical $\pi/8$~gate on a qubit encoded in a 2D color code by switching to a third dimension.  It requires the measurement of high-weight Bell stabilizers that couple pairs of 2D color codes---a requirement that is \emph{not} necessary if one used the related approach of dimensional jumps in gauge color codes~\cite{Bombin14, Bombin15}, wherein the 3D color codes have low-weight, geometrically local stabilizers in three dimensions.

In this section, we show that our stacked code has a key advantage over more standard 3D color codes possessing geometrically local stabilizers, in that it can be arranged in a \emph{two-dimensional} geometry.  For the transformation to and from the stacked code in 2D, we require only geometrically local (in 2D) gauge measurements in the bulk, together with Bell stabilizers measurements along one-dimensional boundaries in this 2D layout.

\begin{figure}
\centering
\begin{subfigure}[t]{0.45\textwidth}
\includegraphics[width=0.8\textwidth]{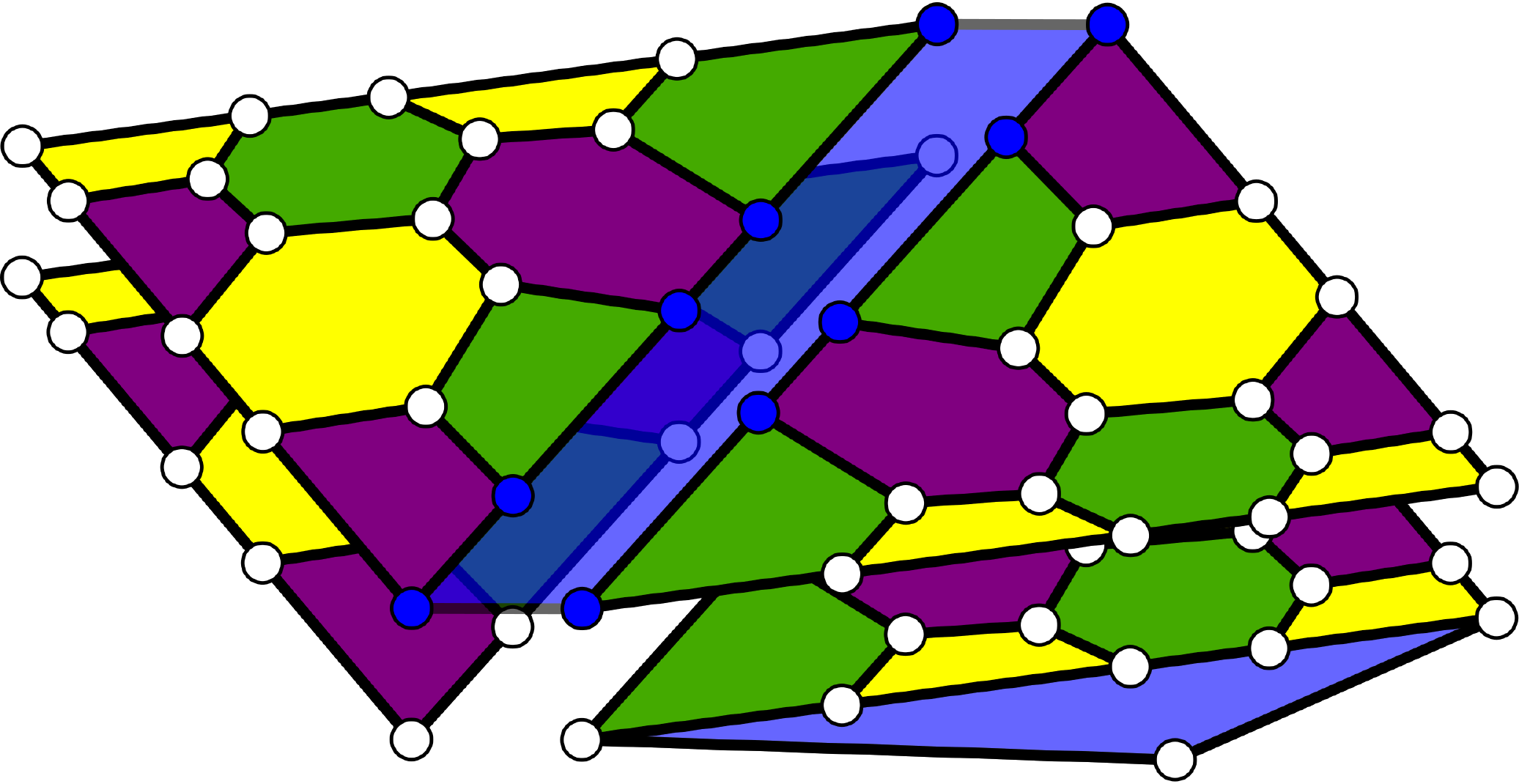}
\caption{}
\label{fig:q2Dinit}
\end{subfigure}
\begin{subfigure}[t]{0.45\textwidth}
\includegraphics[width=0.8\textwidth]{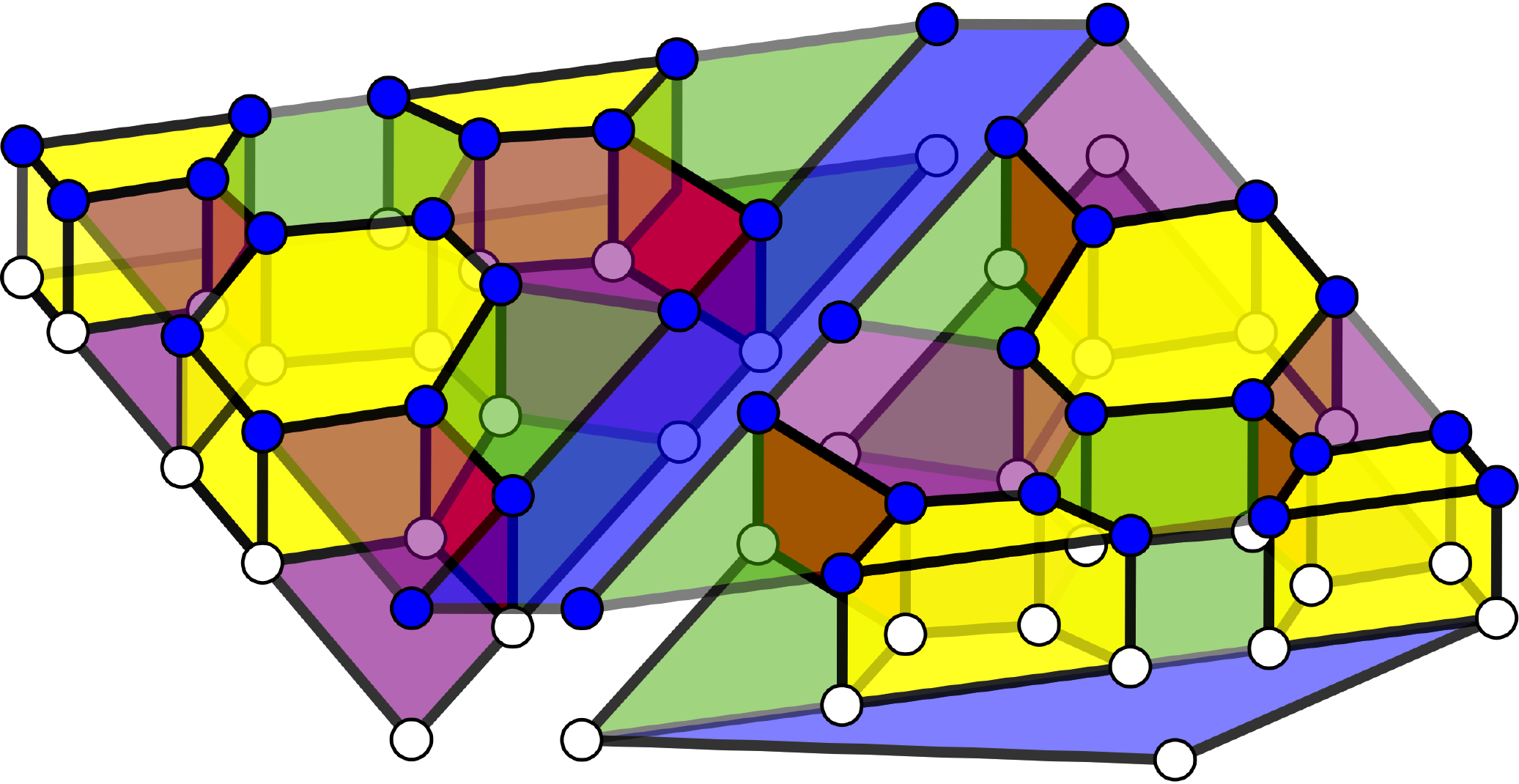}
\caption{}
\label{fig:q2Devolve}
\end{subfigure}
\caption{A 2D layout for the implementation of the stacked code ($d=5$ shown). Pairs of copies of the 2D hexagonal color code are layered on top of one another in a single 2D layer, in such a way as to keep the gauge operators geometrically local. (\subref{fig:q2Dinit}) Initial layout of the stacked code transformation in 2D. The 2D layers $(2k)$ and $(2k+1)$ are coupled by measuring joint logical~$X$ and~$Z$ operators (Bell stabilizers), with supporting qubits shown in blue.  Although Bell stabilizers for the stacked code are high-weight, involving all blue qubits, the only required measurements are those associated with local 2D stabilizer/gauge operators together with one-dimensional operators of weight $\mathcal{O}(d)$~(shaded blue). The only 2D plane that is not initially coupled to another layer (or ancilla qubit) is the bottom $k=1$ layer, which stores the encoded qubit. (\subref{fig:q2Devolve}) Measurement of the weight-4 $Z$-type gauge operators, shown in Red. $X$-type stabilizers from individual layers are combined to form cell-like stabilizers by stabilizer evolution. Original joint logical $X$~measurements, given by Blue shaded region, are mapped to all Blue qubits.}
\label{fig:q2Dlayout}
\end{figure}

\subsection{Arranging the stacked code in two dimensions}
\label{sec:2Dlayout}

Consider the 2D layout of different copies of the 2D~hexagonal color code presented in Fig.~\ref{fig:q2Dlayout}, where layers $(2k-1)$ and $(2k)$ are combined into a single 2D plane and neighboring pairs of layers are arranged next to each other within this 2D plane, equivalent to the doubled color codes of Ref.~\cite{BC15}.  The geometric arrangement can be viewed as unfolding the pairs of copies of the 2D code separated by the Bell stabilizers and tiling the pairs in a 2D plane.  We shall refer to this arrangement as the \emph{unfolded stacked code}.  While it is visually useful to place layers $(2k-1)$ and $(2k)$ separated vertically as in Fig.~\ref{fig:q2Dlayout}, the qubits in these layers can be arranged in a single 2D plane; see Fig.~\ref{fig:local2Dlayout}.

The key feature of this geometric arrangement, which we show in the next section, is that the Bell stabilizers between layers $(2k)$ and $(2k+1)$ can be measured along the shared boundary. Although not geometrically local, this is a very desirable type of measurement from the perspective of physical implementations as the measurement is along a single 1D strip defined by the boundary of the two layers, and may be performed by coupling to a common mode or bus.  One way to ensure fault-tolerance for such a measurement would be to prepare an ancillary state for readout, such as a cat state~\cite{Shor96, DS96}, and repeat the measurement $\mathcal{O}(d)$~times~\cite{AGP06}. The qubits composing the cat state could be arranged along the boundary, and because they will have to be measured to infer the logical measurement, they will be reset and available for the next round of measurement. We note that the scheme is not limited to performing this measurement using a cat state.  Any fault-tolerant readout scheme for these high-weight operators may be applied here, assuming it can conform to the architectural constraints.  We leave this for future work. A nonlocal operation, such as the one described here, is a necessary feature in order to circumvent the Bravyi--K\"onig no-go theorem for constant-depth logical gates outside the Clifford group in topological stabilizer codes in two dimensions~\cite{BK13}.  The resulting code is equivalent to the stacked code, as the joint logical measurement operators along the boundary are mapped to 2D sheets due to the modification of the stabilizers by the gauge measurements.

\subsection{Transformation of the joint boundary Bell stabilizers}

In order to understand the transformation of the Bell stabilizer operators along the boundary, we consider the transformation of stabilizer operators under measurement of anti-commuting Pauli operators.  The $Z$-type Bell stabilizer measurement is straightforward, because the gauge measurements are of type~$Z$ and thus a $Z$-type Bell stabilizer along the boundary remains of that form. This statement is equivalent to the fact that the volume operator of weight~$\mathcal{O}(d^2)$ can be mapped to a boundary plaquette operator due to the gauge~$Z$ measurements. 

Next, we consider the transformation of the joint~$X$ logical boundary operators. Consider an instance of two pairs of 2D codes that are connected by joint logical string operators~$X_{L,s}^{(2k)}X_{L,s}^{(2k+1)}$ and~$Z_{L,s}^{(2k)}Z_{L,s}^{(2k+1)}$, initially shown in Fig.~\ref{fig:q2Dlayout}\subref{fig:q2Dinit}. Let~$\{ H_{c_i}^{(2k-1)}  H_{c_i}^{(2k)} \}$ denote the set of gauge operators that touch the joint logical boundary for 2D layers~$(2k-1)$ and~$(2k)$ of color~$c$, indexed by the label~$c_i$.  Because these~$Z$ operators only intersect with~$X_{L,s}^{(2k)}X_{L,s}^{(2k+1)}$ at a single qubit, these operators anti-commute. Additionally, $\{ H_{c_i}^{(2k-1)} H_{c_i}^{(2k)} \}$ anti-commutes with the individual~$G_{P_{c_i}}$ plaquette operators of matching color from the individual 2D codes $(2k-1)$~and~$(2k)$. The stabilizers of the code are thus modified as follows:  $\{ H_{c_i}^{(2k-1)}  H_{c_i}^{(2k)} \}$ becomes a new stabilizer of the code, replacing~$G_{P_{c_i}}^{(2k)}$. Then, $G_{P_{c_i}}^{(2k-1)}$ is modified by being multiplied by the replaced stabilizer, thus becoming the cell operator~$G_{P_{c_i}}^{(2k-1)}G_{P_{c_i}}^{(2k)}$. Finally, the joint logical operator is also modified by being multiplied by all replaced plaquettes of color~$c$, that is, it becomes
\begin{equation}
  \Bigl( \prod_{c_i} G_{P_{c_i}}^{(2k)} \Bigr) X_{L,s}^{(2k)}X_{L,s}^{(2k+1)}\,.  
\end{equation}
Because similar joint gauge $Z$~measurements are performed between layers~$(2k+1)$ and~$(2k+2)$, the original joint boundary operator is mapped to the operator:
\begin{equation}
  \Bigl( \prod_{c_i} G_{P_{c_i}}^{(2k)}\Bigr) \Bigl( \prod_{c'_j} G_{P_{c'_j}}^{(2k+1)} \Bigr) X_{L,s}^{(2k)}X_{L,s}^{(2k+1)}\,,
\end{equation} 
which corresponds to all qubits on layers~$(2k)$ and $(2k+1)$. An example of the modified joint logical operator is shown in Fig.~\ref{fig:q2Dlayout}\subref{fig:q2Devolve}. The joint logical~$X$ operator is spread over the full 2D lattice, as governed by the transformation of stabilizer operators, and thus becomes one of the blue cells shown in Fig.~\ref{fig:3Dstack}.

\begin{figure}[htbp]
\begin{center}
\includegraphics[width = 0.4\textwidth]{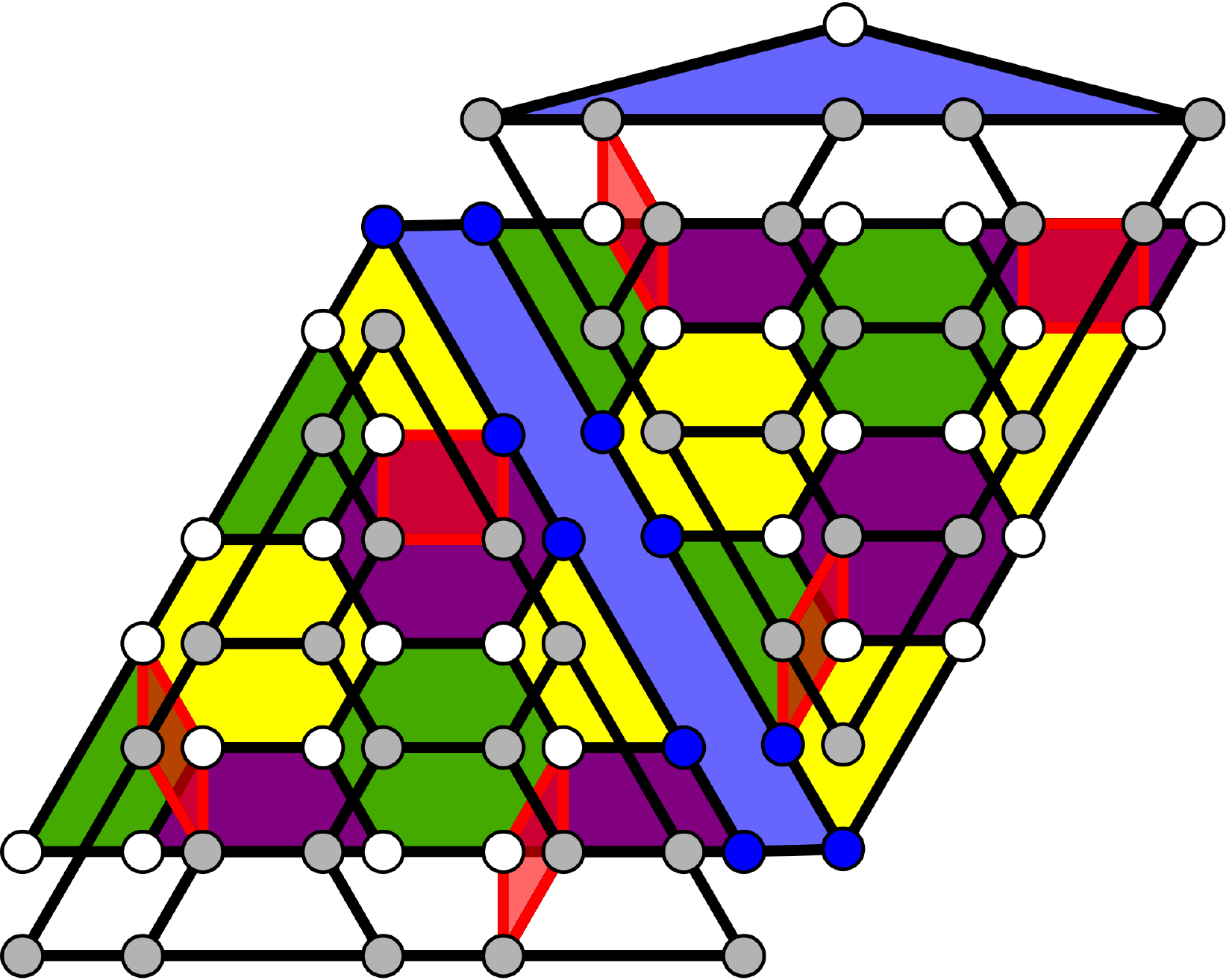}
\caption{A two-dimensional layout of the construction presented in Fig.~\ref{fig:q2Dlayout}. The two originally superimposed lattices have respective grey and white lattice qubits. Only one of the color code stabilizers (per pair) have been colored, for clarity. Gauge measurement operators are given by red faces. Here, we have identified three individual gauge measurements per pair of codes for clarity, there are actually~$3(d^2-1)/8$ such gauge measurements for each pair of distance~$d$ codes.}
\label{fig:local2Dlayout}
\end{center}
\end{figure}

\subsection{Implementation of a fault-tolerant~$\pi/8$~gate in two dimensions}

We now describe how to perform a fault-tolerant $\pi/8$~gate using this stacked code arranged in two dimensions.  We initialize with the information encoded into a 2D color code and pairs of 2D codes laid out edge-to-edge in a 2D arrangement. Bell stabilizers are measured along 1D boundaries between two single sheets from different pairs, before finally measuring out the gauge operators in a local manner between pairs of 2D sheets. Having completed this process, the original information of the 2D code is now stored in a stacked code, and the non-Clifford $\pi/8$~gate can be executed transversally. After completion of the gate, the process is reversed by measuring the original stabilizers of the 2D code and ancilla qubits. The information is mapped back into the 2D color code, where transversal Clifford gates are available for further logical computation. 

We emphasize that the expanded joint logical operators are never  measured in practice, as the transformation from the 2D color code to the stacked code only serves for the application of the logical $\pi/8$  gate. Because the code has distance~$d$ throughout the process without coupling qubits during the measurements, the procedure remains fault-tolerant. If an error of weight less than~$d$ were to occur while the state was encoded in the stacked code, such an error will anticommute with one of the stabilizer cells of the stacked code. We covered the case when it anticommutes with one of the cells of the original 2D code color in Sec.~\ref{sec:FTPi8}. Thus, consider the case where the error anti-commutes with $X_{L,2D}^{(2k-1)}X_{L,2D}^{(2k)}$, where this joint logical operator is across the full 2D surface of the sheets. However, note the following:
\begin{align*}
X_{L,2D}^{(2k-1)}X_{L,2D}^{(2k)} =  \prod_i G_{P_{c_i}}^{(2k-1)} \prod_j G_{P_{c_j}}^{(2k)}  \left( X_{L,b}^{(2k-1)} X_{L,b}^{(2k)}\right),
\end{align*}
where $X_{L,b}^{(2k-1)} X_{L,b}^{(2k)}$ is the joint boundary operator of color~$c$ that is shared by both 2D sheets, and $G_{C_i}^{(l)}$ are the individual $X$~stabilizers of color~$c$ of the two sheets. Therefore, the product of the outcome of all these individual measurements will have to be preserved, that is by taking their product one can infer the measurement outcome of the joint logical operator across the full 2D sheets, as given by the blue qubits in Fig.~\ref{fig:q2Devolve}. As such, this large weight operator does not actually have to be measured to ensure fault-tolerance and rather it is sufficient to measure the individual 2D operators and joint-logical operators along their boundary after the completion of the transversal $\pi/8$~gate. 

This construction results in a fault-tolerant application of a logical~$\pi/8$~gate, yet the growing size of the joint boundary operators leave open the question of whether a rigorous fault-tolerance threshold exists.  We note that, although the subdivision gadget of Ref.~\cite{BC15} establishes a method to reduce the overall weight of the individual operators that have to be measured, it bears similarities to weight reduction techniques proposed in subsystem codes~\cite{Bacon06} which exhibit a decreasing pseudothreshold for each distance rather than a threshold.

\subsection{Comparison to Bravyi--Cross result}

We briefly compare our construction with that of the very recent parallel result by Bravyi and Cross~\cite{BC15}.  In that paper, the authors present a construction of a code for the application of a transversal $\pi/8$~gate through the construction of a triply even code from multiple copies of doubly even codes. They use a construction that mirrors the construction presented here, where 2D color code lattices are chosen with two qubits per site, denoted ``doubled color codes." Each 2D lattice interacts with another 2D lattice through a joint logical measurement at their boundary (the Bell stabilizers presented in our work). A key insight in Ref.~\cite{BC15} is the proposal of a method to measure the Bell stabilizers using only local gauge measurements by applying a ``subdivision gadget." Jones, Brooks, and Harrington recently proposed a similar method for breaking down the measurement of the Bell stabilizers in the construction of triply even codes based on the 2D $[4.8.8]$~color code~\cite{JBH15}. Their construction is inspired by lattice surgery methods for the implementation of joint logical measurements between two copies of 2D color codes~\cite{HFDvM12,LR14}.

Another key contribution of Ref.~\cite{BC15} is the development of an online decoder to handle the transformation of Pauli errors to non-Pauli errors due to the action of the non-Clifford $\pi/8$~gates. Because this gate transforms $X$~errors into a form of correlated~$X$ and $Z$ errors, this can cause difficulties in the decoding of such errors. The authors introduce a Pauli twirling map after the application of the transversal~$\pi/8$ in order to map the original $X$~error to a probabilistic application of $Z$~errors in combination with the original~$X$ error. This twirling map allows for the construction of a maximum likelihood decoder for error correction.  Techniques developed for the purpose of this decoder could potentially be applied to our construction as well.

\section{Mapping from a larger distance 2D color code to the stacked code}
\label{sec:Distances}

\begin{figure}[htbp]
\begin{center}
\includegraphics[width = 0.4\textwidth]{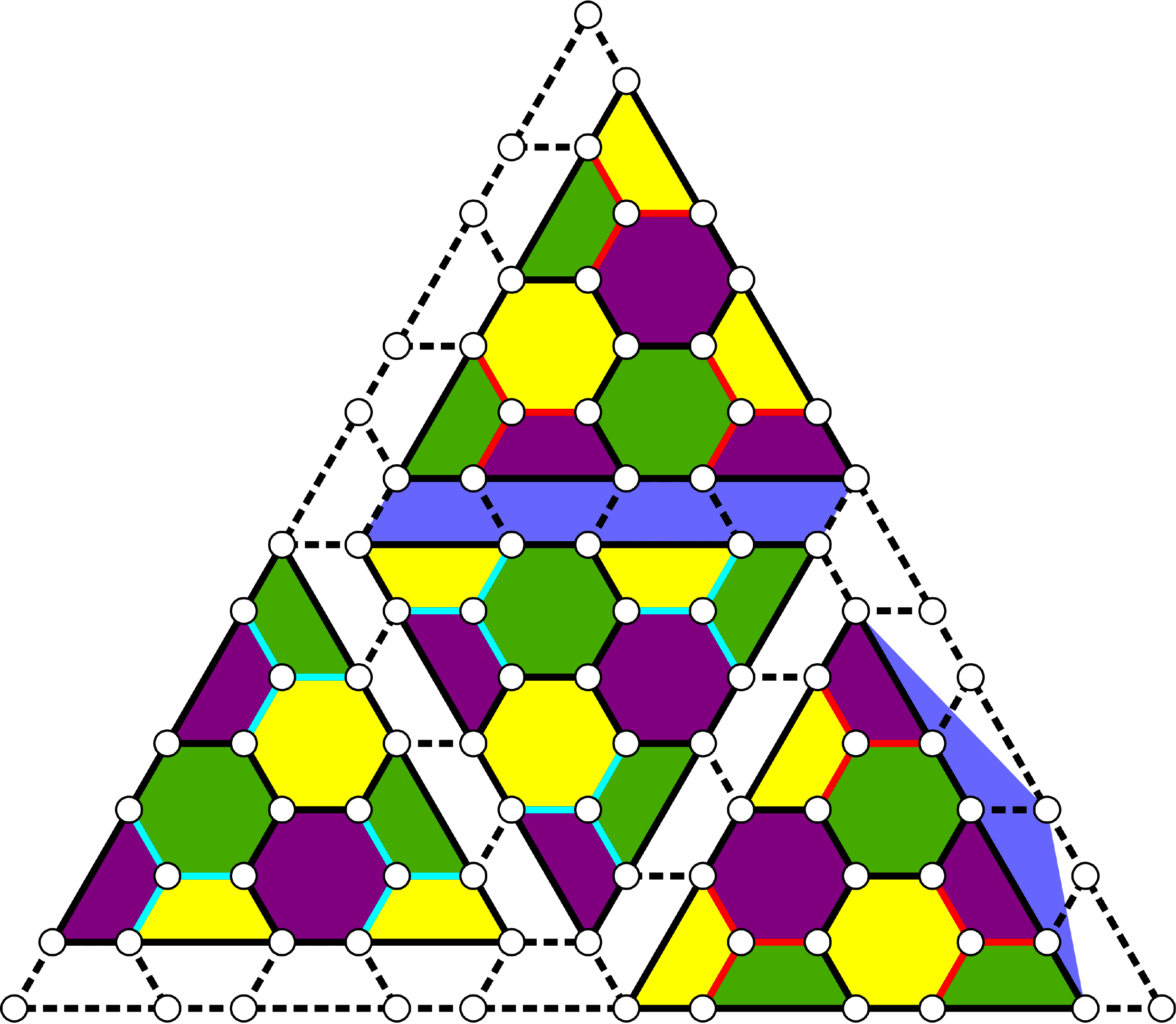}
\caption{Initial coupling of split regions of a 2D color code. The original code is split into multiple color code copies by turning off and modifying certain stabilizer measurements. Different patches are coupled to form a Bell pair by measuring joint logical~$X$ and~$Z$ stabilizers between them, shown in blue. The patch that is not coupled in this way retains the quantum information that was originally stored in the code. The different patches are then further joined together by measuring gauge operators by matching up weight-2 edges from the different patches (forming weight-4 gauge operators), shown by red and cyan edges.}
\label{fig:2Dlayout}
\end{center}
\end{figure}

In this section, we describe a procedure to construct the stacked code as a reduction of a single higher distance 2D color code. This analysis is provided not as a direct means to implement the stacked code in 2D, as we believe the scheme outlined in Sec.~\ref{sec:2Dlayout} is a more practical approach. Rather, we introduce this scheme in order to analyze the scaling of the distance of the stacked code architecture when constructed from a larger 2D color code. The motivation of this analysis is to characterize the degree of nonlocality that is required for stacked codes as a function of the larger 2D distance~$d_2$ in order to implement a non-Clifford transversal gate.

To convert between the 2D architecture and the stacked code architecture, consider initializing a qubit encoded in a higher distance 2D color code, with distance $d_2 \ge d\sqrt{d-1}+1$, where~$d$ is the target distance of the stacked architecture. The initial 2D code is then converted to multiple copies of smaller color codes by ``turning off" certain stabilizers and changing the weighting of others, while simultaneously measuring joint logical~$X$ and~$Z$ operators between neighboring pairs of these newly formed smaller regions, as shown in Fig.~\ref{fig:2Dlayout}.  The logical qubit that was encoded in string operators spanning the full distance of the 2D code is mapped by this process to only a single patch in the 2D layout---the patch that is not paired with another.  This process corresponds to initializing the different layers of the 3D stacked code before the measurement of the gauge operators. The errors that occur can be tracked by recording the statistics of the measurement of the stabilizers before and after their modification, mirroring the technique for various logical gates in the 2D surface code~\cite{FMMC12}. The disadvantage of this architecture is that the gauge operators have to be measured by pairing qubits at different spatial locations in the 2D code, by matching individual edges in each code to form weight-4 operators. A particular set of edges that could be used for gauge measurements in the case of $d=5$ is identified in Fig.~\ref{fig:2Dlayout}. However, the pairings remain relatively local with respect to the 2D code distance as their separation is~$\mathcal{O}(d) = \mathcal{O}(d_2^{2/3})$, which is the same order of nonlocality as the required joint logical measurements. Therefore, by modifying stabilizer measurements and performing joint measurements whose spatial nonlocality is of order~$\mathcal{O}(d_2^{2/3})$ one can logically map a 2D color code to a 3D color code, thus providing the framework to perform a transversal $\pi/8$~gate and enabling a universal set of fault-tolerant gates. 

The distance penalty one pays for such a process is a reduction from $d_2$ to $d_2^{2/3}$, however note that for two color codes with the same number of physical qubits~$n$, the distance of the 2D color code has scaling~$d_2 = \mathcal{O}(\sqrt{n})$ while the 3D color code has scaling~$d = \mathcal{O}(n^{1/3}) = \mathcal{O}(d_2^{2/3})$. Such a distance penalty is to be expected, as the no-go result of Bravyi and K\"onig states that any circuit of depth~$h$ whose individual gates have geometric nonlocal range~$r$ that satisfies~$hr \ll d^{1/2}$, for a 2D topological stabilizer code, can only implement a gate from the Clifford group. Therefore, it should be expected that if one can map to code that can implement a transversal $\pi/8$~gate the degree of geometric nonlocality must be at least of order~$\mathcal{O}(d_2^{1/2})$, which our scheme clearly satisfies (yet does not saturate). Whether there exists methods to implement a fault-tolerant non-Clifford gate in 2D using a reduced degree of nonlocality is an interesting open problem. 

\section{Conclusion}
\label{sec:Conclusion}

Here, we have introduced stacked codes: a class of 3D color codes composed of individual 2D color code layers. We present a method to implement a universal set of logical gates transversally based on a 2D topological stabilizer code. We show that by layering pairs of 2D hexagonal color codes and connecting individual copies of the color code from different pairs through the measurement of nonlocal Bell-like stabilizers, we can then use gauge measurements as proposed in previous works~\cite{PR13, Bombin14, ADP14} to map the logical information initially stored in a 2D color code into a stacked code. This fault-tolerant transformation allows for the application of a transversal gate outside the Clifford group in a 2D layout without having to resort to magic state distillation.  Our proposal circumvents the Bravyi--K\"onig no-go result for transversal non-Clifford gates in 2D stabilizer codes by relying on a realistic form of nonlocal measurements along 1D boundaries in the 2D lattice. 

Due to the growing size of the joint boundary operators, the proposed scheme for fault-tolerant universal computation may not exhibit a threshold in contrast to traditional 3D gauge color codes~\cite{BNB15}. However, even if this were to be the case, it remains of interest to establish the value of the pseudothreshold for low distance realizations of this scheme for the purposes of near-future experiments as well as potential multilayered quantum error correcting architectures, as in Ref.~\cite{CDT09}. Moreover, the stacked codes merit further investigation into their stabilizer measurement properties, because 3D gauge color codes have the capacity for single-shot measurement~\cite{Bombin15b}. Further research into the development of schemes for nonlocal operations to map a 2D stabilizer code to a 3D code, such as the recent proposal in Ref.~\cite{BC15}, could lead to great reductions in architectural complexity and qubit overhead for the implementation of universal fault-tolerant quantum logic.

\section*{Acknowledgements}
T.~J. would like to acknowledge the support from the Vanier--Banting Secretariat and NSERC through the Vanier Canada Graduate Scholarship. T.~J. would like to thank the University of Sydney for their hospitality where part of this work was completed.  This work is supported by the ARC via the Centre of Excellence in Engineered Quantum Systems (EQuS) Project~No.~CE110001013, and the Intelligence Advanced Research Projects Activity (IARPA) Multi-Qubit Coherent Operations Program No.~W911NF-10-1-0330. T.~J. would like to thank Aleksander~Kubica for insightful discussions.

\appendix
\section{Proof of transversal logical $\pi/8$~gate for the stacked code}
\label{app:GateTransversality}

Consider a $[[n,1,d]]$ qubit 2D~color code whose $X$~and~$Z$~generators are labeled by~$\{ G_i^{(1)} \}$ and~$\{ H_i ^{(1)}\}$, respectively, that encodes a single logical qubit. Consider the following basis for the 2D~code based off the CSS~code construction (up to state normalization):
\begin{align*}
\ket{0_{2D}} &= \prod_i (I + G_i^{(1)}) \ket{0}^{\otimes n} = \sum_{\bm{g_x}} \ket{\bm{g_x}}, \\
\ket{1_{2D}} &= X_L^{(1)} \prod_i (I + G_i^{(1)}) \ket{0}^{\otimes n} = \sum_{\bm{g_x}} \ket{\oline{\bm{g_x}}},
\end{align*}
where $X_L^{(1)} = X^{\otimes n}$ is the~$X$ logical operator for the code and $\bm{g_x}$ is an $n$-bit binary vector that lies in the set of vectors generated by the operators~$\{ G_i \}$, and $\oline{\bm{g_x}} = \bm{g_x} \oplus (1, \hdots, 1)$.

Introduce a ($n+1$)-qubit ancillary system in the following state:
\begin{align}
\dfrac{1}{\sqrt{2}} ( \ket{0_{2D}}\ket{0} + \ket{1_{2D}}\ket{1} ).
\end{align}

The stabilizer generators of the original encoded state and the ancillary state thus correspond to:

\begin{widetext}
\vspace{10pt}
\begin{tabular}{p{8cm}||p{8cm}}
\centering 2D code + ancilla Bell state stabilizers
{\begin{align}
& \{ G_{P_i}^{(1)} \} \otimes I^{\otimes n} \otimes I  \nonumber \\    
&\{ H_{P_i}^{(1)} \} \otimes I^{\otimes n} \otimes I \nonumber \\
& I^{\otimes n}  \otimes  G_{P_i}^{(2)}  \otimes I \nonumber \\
& I^{\otimes n} \otimes  H_{P_i}^{(2)}  \otimes I \nonumber \\
&I^{\otimes n} \otimes X_L^{(2)} \otimes X \nonumber \\
&I^{\otimes n} \otimes Z_L^{(2)} \otimes Z \nonumber 
\end{align}}
&
\centering Equivalent stabilizers
{\begin{align}
&\{ G_{P_i}^{(1)} \} \otimes I^{\otimes n} \otimes I \\    
&\{ H_{P_i}^{(1)} \} \otimes I^{\otimes n} \otimes I \\
&\{ G_{P_i}^{(1)}  \otimes  G_{P_i}^{(2)} \} \otimes I \\
&\{ H_{P_i}^{(1)} \otimes  H_{P_i}^{(2)} \} \otimes I  \\
&I^{\otimes n} \otimes X_L^{(2)} \otimes X  \\
&I^{\otimes n} \otimes Z_L^{(2)} \otimes Z 
\end{align} } 
\end{tabular} \\
\end{widetext}
where the right stabilizer generators are equivalent to those on the left by multiplying lines 3~and~4 on the left by lines 1~and~2, respectively (using the notation $\{ G_i^{(1)} \otimes  G_i^{(2)} \}$ to signify that we are multiplying the corresponding $i^{\text{th}}$ stabilizer of each code with one another). Then by following a procedure similar to that proposed by Anderson~\textit{et al.}~\cite{ADP14}, one can replace the $X$~generators from the first line on the right by measuring appropriate gauge~$Z$ operators. to form a new $(2n+1)$-qubit code. The code remains a valid CSS code as the stabilizers all commute and satisfy the requirements of $C_2 \subset C_1$, where $C_1$ is the classical code whose parity check matrix is given by the~$X$ stabilizers~$\{ G_i^{(1)} \otimes  G_i^{(2)} \}$ and $C_2$ is the classical code whose parity check matrix is obtained from the~$Z$ stabilizers.

We proceed to show we can implement a logical gate from~$\CC_3$ transversally. We define an individual $Z$-axis rotation as follows: $Z(\theta) = \text{diag}[1,e^{i\pi\theta}]$. Suppose that the 2D~color code is chosen such that~$U_T = \otimes_{i=1}^n Z(\theta_i) = Z(\bm{\theta})$ implements a logical phase gate~$S_L = \text{diag}[1, i] \in \CC_2$ (the vector~$\bm{\theta}$ represents the individual rotations about the $Z$~axis on the physical qubits forming the quantum code). Note the following observation:
\begin{align}
U_T \ket{0_{2D}} &= Z(\bm{\theta}) \sum_{\bm{g_x}} \ket{\bm{g_x}} = \sum_{\bm{g_x}} e^{i \pi \bm{\theta} \cdot \bm{g_x}} \ket{\bm{g_x}} = \sum_{\bm{g_x}} \ket{\bm{g_x}} \nonumber \\
& \Longrightarrow e^{i \pi \bm{\theta} \cdot \bm{g_x}} = 1, \ \forall \ \bm{g_x}  \nonumber \\
& \Longrightarrow \bm{\theta} \cdot \bm{g_x} = 0 \mod 2, \ \forall \ \bm{g_x}, \label{eq:0stateCond}\\
U_T \ket{1_{2D}} &= Z(\bm{\theta}) \sum_{\bm{g_x}} \ket{\oline{\bm{g_x}}} = \sum_{\bm{g_x}} e^{i \pi \bm{\theta} \cdot \oline{\bm{g_x}}} \ket{\oline{\bm{g_x}}} = e^{i \pi /2}\sum_{\bm{g_x}} \ket{\oline{\bm{g_x}}} \nonumber \\
& \Longrightarrow e^{i \pi \bm{\theta} \cdot \oline{\bm{g_x}}} = e^{i \pi/2}, \ \forall \ \bm{g_x}  \nonumber \\
& \Longrightarrow \bm{\theta} \cdot \oline{\bm{g_x}} = \dfrac{1}{2} \mod 2, \ \forall \ \bm{g_x}. \label{eq:1stateCond}
\end{align}
The assumption that the transversal gate~$U_T$ implements a logical phase gate translates into conditions on the individual physical rotations~$\bm{\theta}$ coupled to the form of the binary vectors~$\bm{g_x}$ related to the $X$~generators of the 2D~quantum code. Consider the CSS~code proposed in Sec.~\ref{sec:TransformingCodes}, where the $X$~generators are given by,
\begin{align*}
\{ G_i^{(1)} &\otimes  G_i^{(2)} \}  \otimes  I \\
I^{\otimes n} &\otimes X_L^{(2)}  \otimes X,
\end{align*}
then a particular choice of code states can be obtained by the CSS~code construction as (upto state normalization):
\begin{align}
\ket{0_{3D}} &= (I^{\otimes (2n+1)}  + I^{\otimes n} \otimes X_L^{(2)}  \otimes X) \\
& \qquad \times \prod_i (I + G_i^{(1)} \otimes  G_i^{(2)} \otimes  I) \ket{0}^{\otimes (2n+1)} \nonumber \\
& = (I^{\otimes (2n+1)}  + I^{\otimes n} \otimes X_L^{(2)}  \otimes X) \sum_{\bm{g_x}} \ket{\bm{g_x}} \ket{\bm{g_x}} \ket{0} \nonumber \\
& = \sum_{\bm{g_x}} \Big(  \ket{\bm{g_x}} \ket{\bm{g_x}} \ket{0} +  \ket{\bm{g_x}} \ket{\oline{\bm{g_x}}} \ket{1} \Big), \\
\ket{1_{3D}} &= (X_L^{(1)} \otimes X_L^{(2)} \otimes X ) \ket{0_{2D}} \nonumber \\
&=  \sum_{\bm{g_x}} \Big(  \ket{\oline{\bm{g_x}}} \ket{\bm{g_x}} \ket{0} +  \ket{\oline{\bm{g_x}}} \ket{\oline{\bm{g_x}}} \ket{1} \Big).
\end{align}

\begin{claim}
\label{clm:TransversalC3}
The~$(2n+1)$~qubit transversal gate $V_T = Z(\frac{\bm{\theta}}{2})  \otimes Z(\frac{\bm{\theta}}{2}) \otimes Z(\alpha)$, where~$\alpha$ is chosen such that $\alpha \in \{ 1/4, 5/4 \}$, implements a logical~$T$ or $TZ$ gate in the logical computational basis~$\{ \ket{0_{3D}}, \ket{1_{3D}} \}$, where $T=\pi/8$~gate.
\end{claim}

\begin{proof}
For the purpose of this proof, we consider the case where the $\pi/8$~gate has the form~$T = \text{diag}[1, e^{i \pi/4}]$, which is equivalent to $\text{diag}[e^{-i \pi/8}, e^{i \pi/8}]$ up to a global phase. Consider first the action of $V_T$ upon the state~$\ket{0_{3D}}$ which should return~$\ket{0_{3D}}$ without a phase.
\begin{align}
V_T \ket{0_{3D}} &= \sum_{\bm{g_x}} \Big( e^{i \pi \frac{\bm{\theta}}{2} \cdot \bm{g_x}} e^{i \pi \frac{\bm{\theta}}{2} \cdot \bm{g_x}} \ket{\bm{g_x}} \ket{\bm{g_x}} \ket{0} \nonumber \\
&\qquad + e^{i \pi \frac{\bm{\theta}}{2} \cdot \bm{g_x}} e^{i \pi \frac{\bm{\theta}}{2} \cdot \oline{\bm{g_x}}} e^{-i \pi \alpha } \ket{\bm{g_x}} \ket{\oline{\bm{g_x}}} \ket{1} \Big) \\
&= \sum_{\bm{g_x}} \Big( e^{i \pi \bm{\theta} \cdot \bm{g_x}} \ket{\bm{g_x}} \ket{\bm{g_x}} \ket{0} \nonumber \\
&\qquad + e^{i \pi \frac{\bm{\theta}}{2} \cdot \bm{g_x}} e^{i \pi \frac{\bm{\theta}}{2} \cdot \oline{\bm{g_x}}} e^{-i \pi \alpha } \ket{\bm{g_x}} \ket{\oline{\bm{g_x}}} \ket{1} \Big) \label{eq:coeff0}\\
&= \sum_{\bm{g_x}} \Big(\ket{\bm{g_x}} \ket{\bm{g_x}} \ket{0} + \ket{\bm{g_x}} \ket{\oline{\bm{g_x}}} \ket{1} \Big), \label{eq:0stateFinal}
\end{align}
where the first coefficient in~\eqref{eq:coeff0} is equal to 1 by the identity in Eq.~\ref{eq:0stateCond}, and the second coefficient is equal to 1 by the following observation. Define the phase~$a$ to be the phase $e^{i \pi a} = e^{i \pi \frac{\bm{\theta}}{2} \cdot \bm{g_x}} e^{i \pi \frac{\bm{\theta}}{2} \cdot \oline{\bm{g_x}}}$. Due to the symmetries of color codes, the value of~$a$ in the following is independent of~$\bm{g_x}$:
\begin{align}
& \dfrac{\bm{\theta}}{2} \cdot \bm{g_x} + \dfrac{\bm{\theta}}{2} \cdot \oline{\bm{g_x}} = a \mod 2 \nonumber \\
\Longrightarrow &\ \bm{\theta} \cdot \bm{g_x} + \bm{\theta} \cdot \oline{\bm{g_x}}  = 2a \mod 2 \nonumber \\
\Longrightarrow &\ 0 + \dfrac{1}{2} = 2a \mod 2 \nonumber \\
\Longrightarrow &\ a = \{ \dfrac{1}{4}, \dfrac{5}{4} \} \mod 2, \label{eq:phase}
\end{align}
therefore $\alpha$ is chosen in order to set the coefficient equal to 1. Consider now the action of~$V_T$, with the appropriate choice of~$\alpha$ for the state~$\ket{1_{3D}}$, which should return the state~$\pm e^{i \pi/4} \ket{1_{3D}}$.
\begin{align}
V_T \ket{1_{3D}} &=  \sum_{\bm{g_x}} \Big( e^{i \pi \frac{\bm{\theta}}{2} \cdot \oline{\bm{g_x}}} e^{i \pi \frac{\bm{\theta}}{2} \cdot \bm{g_x}}  \ket{\oline{\bm{g_x}}} \ket{\bm{g_x}} \ket{0} \nonumber \\
&\qquad + e^{i \pi \frac{\bm{\theta}}{2} \cdot \oline{\bm{g_x}}} e^{i \pi \frac{\bm{\theta}}{2} \cdot \oline{\bm{g_x}}}  e^{-i \pi \alpha } \ket{\oline{\bm{g_x}}} \ket{\oline{\bm{g_x}}} \ket{1} \Big),
\end{align}
which given a choice of~$\alpha$ gives the following:
\begin{align}
V_T \ket{1_{3D}} &= \sum_{\bm{g_x}} \Big( e^{i \pi \alpha} \ket{\oline{\bm{g_x}}} \ket{\bm{g_x}} \ket{0} + e^{i \pi (\frac{1}{2} - \alpha)} \ket{\oline{\bm{g_x}}} \ket{\oline{\bm{g_x}}} \ket{1} \Big), \nonumber \\
& = e^{i \pi \alpha} \ket{1_{3D}} = \pm e^{i \pi/4} \ket{1_{3D}} ,
\end{align}
since $\alpha  = \dfrac{1}{2} - \alpha \mod 2$.
\end{proof}

Therefore we can apply a transversal $\pi/8$~gate to the code construction given above by applying a transversal~logical $Z$ gate at the completion of our gate (the action of $T$ or $TZ$ is fixed by the code and is not probabilistic).

\begin{corollary}
The stacked code has a transversal logical $\pi/8$~gate.
\end{corollary}

\begin{proof}
The only assumption the proof of Claim~\ref{clm:TransversalC3} makes about the ancilla state is that the rotation~$Z(\alpha)$ induces a phase of $e^{i \pi \alpha}$ on the~$\ket{1}$ state and leaves the $\ket{0}$ state invariant. Therefore, we replace the single physical qubit by a logical qubit~$\{ \ket{0_{3D}}, \ket{1_{3D}} \}$ prepared in a 3D state according to the construction laid out in this appendix. Replacing the single~$Z$ of angle~$\alpha$ by a transversal rotation as given by the construction of the previous claim, we can recursively build the stacked code to implement an overall transversal rotation of the $\pi/8$~gate for the stacked code.
\end{proof}

\bibliographystyle{ieeetr}
\bibliography{bibtex_jochym}

\end{document}